\documentclass[journal]{IEEEtran}
\usepackage[ruled]{algorithm2e}
\usepackage{epsfig} % for postscript graphics files
\usepackage{mathrsfs}
\usepackage{verbatim}
\usepackage{graphicx}
\usepackage{float}
\usepackage{graphics} % for pdf, bitmapped graphics files
\usepackage{epsfig} % for postscript graphics files
\usepackage{times} % assumes new font selection scheme installed
\usepackage{amsmath} % assumes amsmath package installed
\usepackage{amssymb}  % assumes amsmath package installed
\usepackage{bm}
\usepackage{cite}
\usepackage{graphicx}
\usepackage[numbers,sort&compress]{natbib}
\usepackage{booktabs}
\usepackage{color}
\usepackage{indentfirst}
\usepackage{soul}
\usepackage{xcolor}
\usepackage{framed}
\usepackage{lipsum}
\usepackage{balance}
\usepackage{multirow}
\colorlet{shadecolor}{yellow!20}

\newtheorem{Proposition}{Proposition}
\newenvironment{proof}{{\emph {Proof.}}}{\hfill $\square$\par}

\ifCLASSINFOpdf
\else
\fi
\begin{document}
\title{Minimum Snap Trajectory Generation and Control for an Under-actuated Flapping Wing Aerial Vehicle}

\author{Chen~Qian,
        Rui~Chen,
        Peiyao~Shen,
        Yongchun~Fang,
        Jifu~Yan,              
        and Tiefeng~Li
\thanks{This work was supported by China Postdoctoral Science Foundation Funded Project (Project No.:2022M722912),

C. Qian is with Interdisciplinary Innovation Research Centers, Intelligent Robotic Research Center, Zhejiang Laboratory, Hangzhou
311100, China (e-mail: qianc@zhejianglab.com).

R. Chen and T. Li are with School of Aeronautics and Astronautics, Zhejiang University, Hangzhou 310012, China (e-mail: 22124067@zju.edu.cn, litiefeng@zju.edu.cn).
T. Li is the corresponding author.

P. Shen is with Shanghai NIO Automobile Co., Ltd., Shanghai, 201804, China (e-mail: shenpy3631@outlook.com).

Y. Fang and J. Yan are with College of Artificial Intelligence, Nankai University, and Institute of Robotics and Automatic Information
Systems, Nankai University, Tianjin, 300353, China (e-mail: fangyc@nankai.edu.cn, yanjf2000@mail.nankai.edu.cn).

% The authors are with the Institute of Robotics and Automatic Information
% Systems, College of Artificial Intelligence, and also with the
% Tianjin Key Laboratory of Intelligent Robotics, Nankai University, Tianjin
% 300353, China (e-mail: chainplain@mail.nankai.edu.cn; fangyc@nankai.edu.cn).
}
}


\markboth{PREPARING SUBMISSION, AUG 28TH 2023}%
{Shell \MakeLowercase{\textit{et al.}}: Bare Demo of IEEEtran.cls for IEEE Journals}



\maketitle

\begin{abstract}
This paper presents both the trajectory generation and tracking control strategies for an underactuated flapping wing aerial vehicle (FWAV).
First, the FWAV dynamics is analyzed in a practical perspective.
Then, based on these analyses, we demonstrate the differential flatness of the FWAV system, and develop a general-purpose trajectory generation
strategy.
Subsequently, the trajectory tracking controller is developed with the help of robust control and switch control techniques.
After that, the overall system asymptotic stability is guaranteed by Lyapunov stability analysis.
To make the controller applicable in real flight, we also provide several instructions.
Finally, a series of experiment results manifest the successful implementation of the proposed trajectory generation strategy and tracking control strategy.
This work firstly achieves the closed-loop integration of  trajectory generation and control for real 3-dimensional flight of an underactuated FWAV to a practical level.
\end{abstract}


% \fi
\IEEEpeerreviewmaketitle

\section{Introduction}
\IEEEPARstart{F}{lapping} wing flight, much like a midair acrobatic skill defying the grasp of gravity, bestows upon the flier a realm of unparalleled maneuverability and agility. Just as birds and insects effortlessly navigate the vast expanse of the sky, flapping wings empower robots to transcend the limitations of traditional fixed-wing or rotary-wing systems and gracefully fly.
These advantages have led to increasing interest in developing flapping wing systems for applications such as aerial surveillance, environmental monitoring, as well as search and rescue operations \cite{Hassanalian-2017, Decroon-2020, Wang-2022}.
However, due to the complex dynamics, achieving stable and efficient flight in flapping wing systems presents significant challenges \cite{Sane-2002, Chin-2016,Helbling-2018, Sihite-2020}.
To overcome these challenges, much endeavor has been devoted to modeling that are oriented towards real-world flapping wing flight missions \cite{Khan-2021, Wangs-2022,McGill-2022,Biswal-2019}.
These endeavors typically focus on a particular flight objective, such as regulating attitude or devising optimal flight paths. Consequently, it is crucial to develop a methodology that adequately addresses the intricacies arising from the underactuated and nonlinear characteristics inherent in flapping wings dynamics, which should be compatible with both trajectory generation and tracking requirements.

Trajectory generation, in conjunction with compatible trajectory tracking control, can be regarded as a pivotal challenge that significantly impedes the practical applicability of flapping wing robots.
However, the researches on the trajectory generation method for flapping wing flight are relatively rare comparing with those on the fixed-wing flight or the rotary-wing flight \cite{Hoff-2019}.
Several representative works are discussed below to provide a comprehensive overview of the current state of the field.
In \cite{Paranjape-2013}, the perching maneuver is achieved by the aerial robot through a combination of wing articulation and control algorithms. The algorithms use closed-loop motion planning and dynamic inversion techniques to ensure stability and precise control during the perching maneuver. The robot executes a pitch up with maximum upward elevator deflection to achieve rapid deceleration and flatten the flight path, leading to a successful perched landing. 
In order to make the Bat Bot (a bat-like flapping wing robot) to navigate and perform various tasks in shared environments, the authors of \cite{Hoff-2019}
propose a generalized approach that uses a model with direct collocation methods to plan dynamically feasible flight maneuvers.
Then in \cite{Hoff-2021}, a two-stage optimization routine to plan flapping flight trajectories is proposed. In order to achieve minimum effort spent
moving the hind limbs of the Bat Bot, they firstly use the fixed wing model to solve the optimization 
problem, then use this result as the initial guess and
subsequently use the flapping wing model.
In \cite{rod-2022}, the authors aim at high endurance flight of flapping wing robots. They propose a graph-based approach that builds a tree to search for dynamically feasible and energy-efficient trajectories. 

Although these trajectory generation approaches make innovative probes, they are regretfully limited for specific flapping wing robots or application scenario, such that many intriguing problems remain unveiled.
Focused on gliding maneuvering, the flapping wing dynamics are unsurprisingly not considered in \cite{Paranjape-2013}.
And the planning method provided in \cite{Hoff-2019} depends heavily on the the load cell data, which they use for selecting model parameters to improve modeling accuracy. 
Neither the strategy in \cite{Hoff-2019} nor in
\cite{Hoff-2021} incorporates the trajectory 
generation method with tracking control. 
This open loop flight fashion makes the flapping wing aerial vehicles (FWAVs) 
prone to external disturbances, and thus makes
them less practically applicable, especially for outdoor tasks.
Moreover, only 2D trajectories are considered in \cite{rod-2022}, which is therefore not suitable for many flapping wing robot practical applications. Furthermore, the obstacle avoidance problem is rarely considered. However, the ability to navigate safely and autonomously in complex environments is crucial for the practical applications of aerial vehicles  \cite{Kong-2021, Tijmons-2017, Ol-2008}. 

Once the desired trajectory is generated, the trajectory tracking control comes into play.
Controller proposed in \cite{Wissa-2020} includes an integral sliding mode control law that manipulates the FWMAV dynamics and provides robust performance against model uncertainties and external disturbance.
In \cite{He-2020}, a neural network based controller
is proposed with accurate trajectory tracking on the vertical plane. 
And a vector field based trajectory tracking controller is proposed in \cite{Ndoye-2023}, which is robust to various initial position and velocity conditions. In \cite{Fei-2023}, the authors realize a nonlinear flight controller onboard, which incorporates
parameter adaptation and robust control technique.
The controller developed by them effectively tackles a range of issues, including: mitigating the sensing challenges arising from significant oscillations, accounting for the highly nonlinear and unsteady aerodynamics associated with flapping wing motion, accommodating uncertainties in system parameters, and counteracting external disturbances.

The absence of real flight experimentation for the first two tracking controllers \cite{Wissa-2020,He-2020} unavoidably dilutes the persuasiveness of their results. 
Furthermore, the underactuated nature of specific FWAVs is not comprehensively addressed in the aforementioned works.
And, the conventional vector-aided controller \cite{Ndoye-2023} commonly employed in path-following problems \cite{Nelson-2007, Zhao-2018} fails to fulfill specific position requirements at precise time instances, primarily due to the independence between initial conditions and time. 
While this independence is often regarded as advantageous for enhancing robustness, it hinders the control ability to meet specific position and velocity requirements within predefined temporal constraints.

Based on the observations above, the key contributions of this work can be concluded into the following three folds:
\begin{enumerate}
  \item The theoretical bases for the trajectory planning of the studied underactuated FWAV are established.
  \item A novel trajectory tracking controller, which is compatible with the trajectory generation strategy, is proposed in this study. 
  \item To the best of our knowledge, this study presents the first successful closed-loop integration of trajectory generation and control for real 3-dimensional flight of an underactuated FWAV.
\end{enumerate}

The remainder of this paper is organized as follows.
In section~\uppercase\expandafter{\romannumeral2}, we analyze the FWAV dynamics for further planning and control tasks.
Then the trajectory generation strategy and the trajectory tracking control strategy are presented in section~\uppercase\expandafter{\romannumeral3} and section~\uppercase\expandafter{\romannumeral4}, respectively.
After that, the real flight experiment results are provided and analyzed in section~\uppercase\expandafter{\romannumeral5}.
Finally, we conclude this work in section~\uppercase\expandafter{\romannumeral6}.

\section{Dynamics Analysis}
\subsection{Flapping Wing Dynamics}
Unit-quaternion has several advantages for representing attitude, including efficient interpolation between orientations, avoiding singularities (such as gimbal lock), and providing a compact representation.
The unit-quaternion $\bm q \in \mathcal S^3$ maps an attitude onto two elements on $\mathcal S^3$, which consists of a scalar part and a vector part, that is $\bm q = {\left[ {\begin{array}{*{20}{c}} \eta &{{\bm \epsilon ^\top}}\end{array}} \right]^\top}$. 
The conjugate quaternion of $\bm{q}$ represents as $\bm{q}^*={\left[ {\begin{array}{*{20}{c}} \eta &{-{\bm \epsilon ^\top}}\end{array}} \right]^\top}$, and  $R\left( \bm q \right) = R^\top\left( {  {\bm q}^*} \right)$.
Moreover, the rotation matrix can be expressed as 
\begin{align}
\label{eq:q2R}
R\left( \bm q \right) = I + 2\eta {\left[ \bm \epsilon  \right]_ \times } + 2\left[ \bm \epsilon  \right]_ \times ^2
\end{align}
where $I \in \mathbb R^{3 \times 3}$ is the identity matrix, and $\left[ \bm \star  \right]_ \times$ is the skew-symmetric matrix of $\bm \star \in \mathbb R^3$, such that 
\begin{align}
\label{eq:skew-symmetric}
{\left[ \bm \epsilon  \right]_ \times } = \left[ {\begin{array}{*{20}{c}}
0&{ - {\epsilon _3}}&{{\epsilon _2}}\\
{{\epsilon _3}}&0&{ - {\epsilon _1}}\\
{ - {\epsilon _2}}&{{\epsilon _1}}&0
\end{array}} \right]
\end{align}

Based on our previous study in \cite{Pre1}, the FWAV can be modeled by the following dynamics:
\begin{align}
\label{eq:ConDynamics}
\frac{{\rm{d}}}{{{\rm{d}}t}}\left[ {\begin{array}{*{20}{c}}
 \bm{p}\\
\bm v\\
 \bm{q}\\
{ \bm{\omega} }
\end{array}} \right] &= \left[ {\begin{array}{*{20}{c}}
{\bm v}\\
{-g{\bm{e}_3} +{\bm{q} \otimes } \frac{{{\bm{F}_{\rm drag}}}} {m} \otimes\bm{q}^*}\\
{\frac{1}{2}\bm{q}{ \otimes } {\bm{\omega}}} \\
{ - {{J}^{ - 1}}{(}\bm{\omega}  \times {{J} } \bm{\omega} )}
\end{array}} \right]\nonumber\\
 &+ \left[ {\begin{array}{*{20}{c}}
0\\
{\bm{q} \otimes \frac{{F}_{\rm thrust}\bm e_3} {m} \otimes {\bm{q}^*}}\\
0\\
{{{J}^{ - 1}}{\bm{\tau} _{\theta}}}
\end{array}} \right]
\end{align}
where the FWAV dynamics is modeled with the following states ${\left[ {\begin{array}{*{20}{c}}
{{\bm p^\top}}&{{\bm v^\top}}&{{\bm q^\top}}&{{\bm \omega ^\top}}&{{f_{\rm flap}}}&{{\theta _{\rm rud}}}&{{\theta _{\rm ele}}}
\end{array}} \right]^\top} \in \mathbb R^{16}$, 
which consists of the position $\bm p \in \mathbb R^3$, the velocity $\bm v \in \mathbb R^3$, 
the unit-quaternion $\bm q \in \mathcal S^3$, the angular velocity $\bm \omega$,
the flapping wing frequency ${f_{\rm flap}} \in \mathbb R^+$, 
the rudder deflection angle ${\theta_{\rm rud}} \in \mathbb R$,
and the elevator deflection angle ${\theta_{\rm ele}} \in \mathbb R$.
Both  $\bm p$ and $\bm v$ are represented in the inertia frame.
The symbol $g \in \mathbb R^+$ is the gravitational acceleration, 
and symbol $m \in \mathbb R^+$ is the mass of the FWAV.
The matrix $J \in \mathbb R ^{3 \times 3}$ is the inertia matrix. 
The vector $\bm F _{\rm drag} \in \mathbb R ^ 3$ is the aerodynamic drag force.
The scalar $ F _{\rm thrust} \in \mathbb R ^ +$ is the magnitude of the thrust force, meanwhile,
$\bm e_3 = {\left[ {\begin{array}{*{20}{c}}0&0&1\end{array}} \right]^\top} \in \mathbb R ^ 3$ is a unit vector.
Based on the flapping wing aerodynamics induced in \cite{Sane-2002}, 
we can conclude that the thrust can be modeled as
\begin{equation}
\label{eq:th-vs-flap}
{F_{\rm thrust}} = {k_{\rm tf}}f_{\rm flap}^2
\end{equation}
where $k_{\rm tf} \in \mathbb R^+$ is the thrust coefficient. 
Similarly, the drag force can be modeled as 
\begin{equation}
\label{eq:v-vs-drag}
{\bm F_{\rm drag}} = \left[ {\begin{array}{*{20}{c}}
{ - {k_{{\rm{d,x}}}}{\mathop{\rm sgn}} \left( {{}^B{v_x}} \right){}^Bv_x^2}\\
{ - {k_{{\rm{d,y}}}}{\mathop{\rm sgn}} \left( {{}^B{v_y}} \right){}^Bv_y^2}\\
{ - {k_{{\rm{d,z}}}}{\mathop{\rm sgn}} \left( {{}^B{v_z}} \right){}^Bv_z^2}
\end{array}} \right]
\end{equation}
where the positive constants ${k_{{\rm{d,x}}}}, {k_{{\rm{d,y}}}}, {k_{{\rm{d,z}}}} \in \mathbb R^+$ are the drag coefficients, $^Bv_x$, $^Bv_y$, $^Bv_z\in \mathbb R$ are the FWAV translational velocity along the three axes of the body fixed frame, respectively. Moreover, ${\rm sgn}(\star) \in \mathbb R \to \{-1,1\}$ is the signum function.

When the deflection angle is constrained within a relatively small range, the torque produced $\bm \tau _\theta \in \mathbb R^3$ can be represented by the following mathematical expression:
\begin{equation}
\label{eq:tau_theta}
{\bm \tau _\theta } = \left[ {\begin{array}{*{20}{c}}
{ - \left( {{k_{{\rm{\tau ,x}}}}{\mathop{\rm sgn}} \left( {{}^B{v_z}} \right){}^Bv_x^2 + {k_{{\rm{flap,x}}}}f_{{\rm{flap}}}^2} \right){\theta _{{\rm{rud}}}}}\\
{ - \left( {{k_{{\rm{\tau ,y}}}}{\mathop{\rm sgn}} \left( {{}^B{v_z}} \right){}^Bv_x^2 + {k_{{\rm{flap,y}}}}f_{{\rm{flap}}}^2} \right){\theta _{{\rm{ele}}}}}\\
{ - \left( {{k_{{\rm{\tau ,z}}}}{\mathop{\rm sgn}} \left( {{}^B{v_z}} \right){}^Bv_x^2 + {k_{{\rm{flap,z}}}}f_{{\rm{flap}}}^2} \right){\theta _{{\rm{rud}}}}}
\end{array}} \right]
\end{equation}
where ${k_{{\rm{\tau,x}}}}$, ${k_{{\rm{\tau,y}}}}$, ${k_{{\rm{\tau,z}}}} \in \mathbb R$ are the velocity-induced torque coefficients, and ${k_{{\rm{flap,x}}}}$, ${k_{{\rm{flap,y}}}}$, ${k_{{\rm{flap,z}}}} \in \mathbb R$ are the flapping-wing-induced torque coefficients. 

Furthermore, in accordance with the simplification presented in \cite{Penicka-2022},
the flapping wing frequency and the deflection angles are modeled as first-order systems:
\begin{align}
\label{eq:firstflap}
&{{\dot f}_{{\rm{flap}}}} = {{\left( {{f_{{\rm{flap,c}}}} - {f_{{\rm{flap}}}}} \right)} \mathord{\left/
 {\vphantom {{\left( {{f_{{\rm{flap,c}}}} - {f_{{\rm{flap}}}}} \right)} {{k_{{\rm{flap,c}}}}}}} \right.
 \kern-\nulldelimiterspace} {{k_{{\rm{flap,c}}}}}}\\
 \label{eq:thetarud}
&{{\dot \theta }_{{\rm{rud}}}} = {{\left( {{\theta _{{\rm{rud,c}}}} - {\theta _{{\rm{rud}}}}} \right)} \mathord{\left/
 {\vphantom {{\left( {{\theta _{{\rm{rud,c}}}} - {\theta _{{\rm{rud}}}}} \right)} {{k_{{\rm{rud,c}}}}}}} \right.
 \kern-\nulldelimiterspace} {{k_{{\rm{rud,c}}}}}}\\
  \label{eq:thetaele}
&{{\dot \theta }_{{\rm{ele}}}} = {{\left( {{\theta _{{\rm{ele,c}}}} - {\theta _{{\rm{ele}}}}} \right)} \mathord{\left/
 {\vphantom {{\left( {{\theta _{{\rm{ele,c}}}} - {\theta _{{\rm{ele}}}}} \right)} {{k_{{\rm{ele,c}}}}}}} \right.
 \kern-\nulldelimiterspace} {{k_{{\rm{ele,c}}}}}}
\end{align}
where ${f_{{\rm{flap,c}}}}$, ${\theta _{{\rm{rud,c}}}}$, ${\theta _{{\rm{ele,c}}}} \in \mathbb R$ are the commanded inputs, and the ${k_{{\rm{flap,c}}}}$, ${k _{{\rm{rud,c}}}}$, ${k _{{\rm{ele,c}}}} \in \mathbb R^+$ are the corresponding time constants.

\begin{figure}[t]
\centering
\includegraphics[width=2.8in]{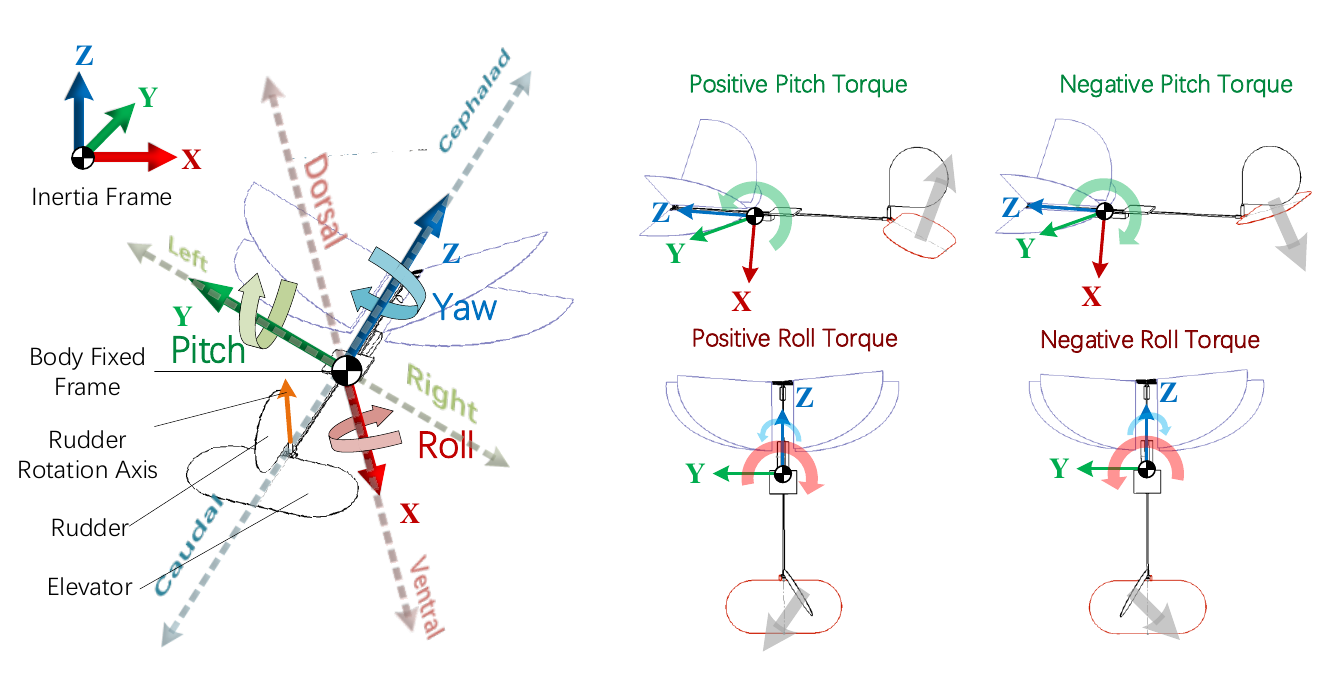}
%ReviseAttack.eps
\caption{Schematic of the developed FWAV: three-dimensional body fixed frame consists of three orthogonal axes, where the X-axis is red, Y-axis green, Z-axis blue.
The gray arrows are the average forces generated by the tails.
When the elevator undergoes a rotation towards the ventral side, the robot manifests a positive pitch torque. Conversely, a rotation towards the dorsal side of the elevator results in the generation of negative pitch torque.
Upon the rotation of the rudder towards the right side, a positive roll torque is engendered, accompanied by a relatively smaller magnitude of negative yaw torque, primarily attributed to a shorter force arm.}
\label{figure:Frames}
\end{figure}

\subsection{Vertical Frame Dynamics}

Firstly, due to the dorsal installation of the rudder, concomitant with the generation of roll torque, it possesses the capability to produce yaw torque. Despite the relatively modest magnitude of the yaw torque, it has the potential to exert an influence on the rotation of FWAV, which is shown in Fig. \ref{figure:Frames}.

Secondly, the wind-vane-like dynamics can also generate yaw torque \cite{Pre1}. 
The average positions of the flapping wings are shown as a ``V'' shape.
To this end, when in normal forward flight, 
the generated yaw torque intends to maintain the flight velocity within or in close proximity to the $X$-$Z$ plane by effecting a rotation of the robot in the yaw direction, which will be discussed in Section~\uppercase\expandafter{\romannumeral 4}.

\begin{figure}[t]
\centering
\includegraphics[width=2.6in]{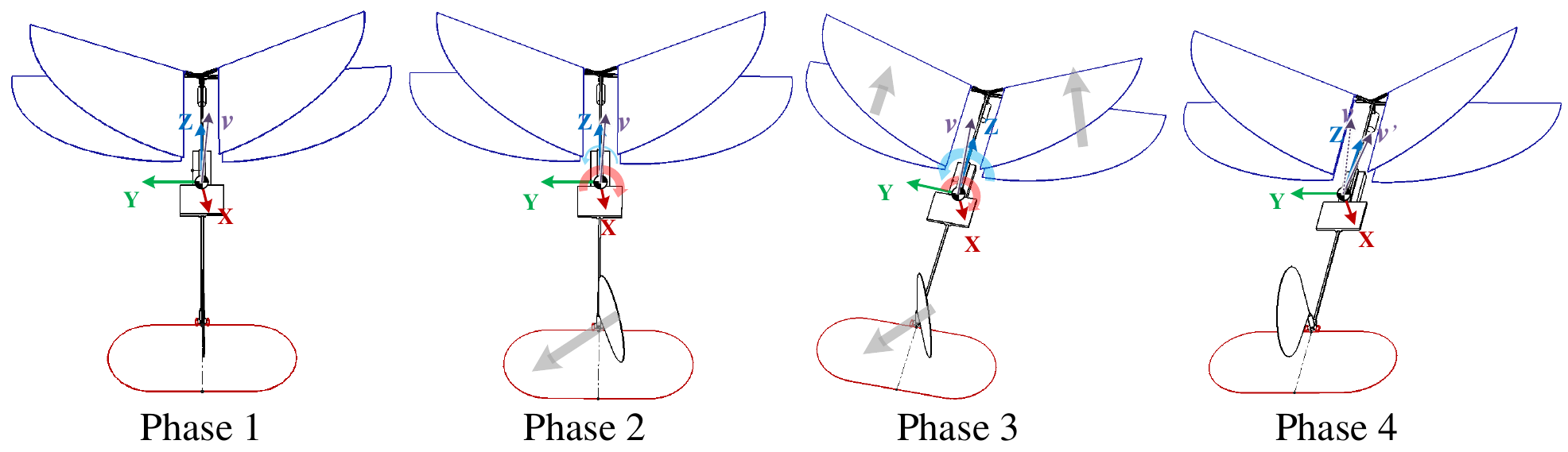}
%ReviseAttack.eps
\caption{Rudder deflection effects of the developed FWAV: the right rotation is demonstrated in this scenario, 
which is divided into 4 phases. We adhere to the axes notation presented in Fig. \ref{figure:Frames}. The translational velocities before and after are denoted as $\bm v$ and $\bm v'$ with color of violet.} 
\label{figure:Rotate}
\end{figure}

The scenario wherein the robot executes a rotation of its rudder during forward flight is demonstrated
in Fig. \ref{figure:Rotate}.
In phase 1, the FWAV exhibits forward flight with a velocity situated within the $X$-$Z$ plane of the body-fixed frame.
During phase 2, the rudder generates both roll torque and yaw torque. The roll torque predominates, resulting in faster manifestation of its effects and leading to phase 2.
Phase 3 marks the emergence of wind-vane-like dynamics as the predominant factor influencing the flight behavior of the FWAV. This dominance initiates rotation around the $Z$-axis of the body-fixed frame.
Subsequently, in phase 4, the rudder rotates back to its neutral position, and the yaw torque generated by the wind-vane-like dynamics causes the $Y$-axis of the body-fixed frame rotates back to the horizontal plane, that is the $X$-$Y$ plane of the inertia frame.
% After this transition, the rudder rotates back.
The cumulative effect of these sequential behaviors culminates in rotation around the $Z$-axis of the inertia frame, and in the meantime, the translational velocity $\bm v$ undergoes a gradual transition to  $\bm v'$.

Based on this observation, using a novel frame called the vertical frame introduced in \cite{Qian-2022} and the dynamics model in \eqref{eq:ConDynamics},
% This frame is denoted as the non-holonomic frame.
the simplified dynamics of the FWAV in the vertical frame is give by
\begin{align}
 &\dot {\bm p} = {\bm v},\label{eq:vertical1}\\
 &{\bm v} = R\left(  \psi  \right){}^V\bm v,\label{eq:vertical2}\\
 &{{}^V{\dot v}_x} = {}^V F_{\rm th,x}/m+ {}^V F_{\rm d,x}/m-\omega_\psi {}^Vv_y,\label{eq:vertical3}\\
 &{{}^V{\dot v}_y} = {}^V F_{\rm th,y}/m+ {}^V F_{\rm d,y}/m+\omega_\psi {}^Vv_x,\label{eq:vertical4}\\
 &{{}^V{\dot v}_z} = {}^V F_{\rm th,z}/m + {}^V F_{\rm d,z}/m -g,\label{eq:vertical5}\\
 &\dot \psi  = \omega_\psi, \label{eq:vertical6}\\
 &\dot\omega_\psi  = {}^V a_{\rm \tau,z} \label{eq:vertical7}
 \end{align}
where $^V\bm v = {\left[ {\begin{array}{*{20}{c}}
{{}^V{v_x}}&{{}^V{v_y}}&{{}^V{v_x}}
\end{array}} \right]^\top} \in \mathbb R^3$ are the translational velocity in the vertical frame,
${}^V F_{\rm th,x} \in \mathbb R$, ${}^V F_{\rm th,y} \in \mathbb R$, and ${}^V F_{\rm th,z} \in \mathbb R$ are the thrust components along the three axes of the vertical frame,
and ${}^V F_{\rm d,x} \in \mathbb R$, ${}^V F_{\rm d,y} \in \mathbb R$, and ${}^V F_{\rm d,z} \in \mathbb R$ are the corresponding drag components.
Moreover, $\psi \in \mathbb R$ is the angle of the rotation from the inertia frame to the vertical frame, namely the azimuth angle,
while $\omega_\psi \in \mathbb R$ is its angular velocity, and $R(\psi) \in {\rm SO(3)}$ is the corresponding 
rotation matrix, which is shown as
\begin{align}
R\left( \psi  \right) = \left[ {\begin{array}{*{20}{c}}
{\cos \psi }&{ - \sin \psi }&0\\
{\sin \psi }&{\cos \psi }&0\\
0&0&1
\end{array}} \right] \nonumber
\end{align}
Finally, ${}^V a_{\rm \tau,z} \in \mathbb R$ is the angular acceleration generated by the rudder, as well as the 
wind-vane-like dynamics. 

The term $R^\top\left( \bm q \right){\bm e_3}$, which is invariant of yaw motion, is denoted as $\bm \Gamma  = {\left[ {\begin{array}{*{20}{c}}
{{\Gamma _x}}&{{\Gamma _y}}&{{\Gamma _z}}
\end{array}} \right]^\top} \in \mathbb R^3$, which is regarded as the reduced attitude, and the corresponding reduced attitude control problem is extensively studied in \cite{Pre1}.

The components of the thrust along the $X$-axis of the vertical frame, and the one along the $Z$-axis of the vertical frame are 
\begin{align}
\label{eq:thrustx}
{}^V{F_{{\rm{th,x}}}} = -{k_{{\rm{tf}}}}f_{{\rm{flap}}}^2\Gamma_x\\
% \label{eq:thrusty}
% {}^V{F_{{\rm{th,y}}}} = {k_{{\rm{tf}}}}f_{{\rm{flap}}}^2\Gamma_y\\
\label{eq:thrustz}
{}^V{F_{{\rm{th,z}}}} = {k_{{\rm{tf}}}}f_{{\rm{flap}}}^2\Gamma_z
\end{align}
where $k_{\rm tf} \in \mathbb R^+$ is the thrust force coefficient, which can be straightforwardly identified in hovering flight.
% where $\bm e_3 = {\left[ {\begin{array}{*{20}{c}}1&0&0\end{array}} \right]^\top} \in \mathbb R ^ 3$ is a unit vector.
% It is noteworthy that, ${}^V{F_{{\rm{th,z}}}}$ is neglected due to the loosely non-holonomic constraint.

Given the intended stability of the FWAV forward flight, it can be observed that the drag force in the vertical frame exhibits a comparable pattern to \eqref{eq:v-vs-drag}, which is given by
\begin{equation}
\left[ {\begin{array}{*{20}{c}}
{{}^V{F_{{\rm{d,x}}}}}\\
{{}^V{F_{{\rm{d,y}}}}}\\
{{}^V{F_{{\rm{d,z}}}}}
\end{array}} \right] = \left[ {\begin{array}{*{20}{c}}
{ - {}^V{k_{{\rm{d}},{\rm{x}}}}{\rm{sgn}}\left( {{}^V{v_x}} \right){}^Vv_x^2}\\
{ - {}^V{k_{{\rm{d}},{\rm{y}}}}{\rm{sgn}}\left( {{}^V{v_y}} \right){}^Vv_y^2}\\
{ - {}^V{k_{{\rm{d}},{\rm{z}}}}{\rm{sgn}}\left( {{}^V{v_z}} \right){}^Vv_z^2}
\end{array}} \right]
\end{equation}
where the positive constants ${{}^Vk_{{\rm{d,x}}}}, {{}^Vk_{{\rm{d,y}}}}, {{}^Vk_{{\rm{d,z}}}} \in \mathbb R^+$ are the drag coefficients in the vertical frame, which can be identified in uniform speed forward flight.

Based on the torque model generated by the rudder \eqref{eq:tau_theta}, there exist the following equations:
\begin{align}
{}^V{a_{{\rm{\tau ,z}}}} = &\underbrace { - \left( {{}^V{k_{{\rm{\tau ,x}}}}{\mathop{\rm sgn}} \left( {{}^V{v_z}} \right){}^Vv_z^2 + {}^V{k_{{\rm{flap,x}}}}f_{{\rm{flap}}}^2} {\Gamma _z}\right){\theta _{{\rm{rud}}}}}_{{\rm{Rudder~induced~yaw~torque}}}\nonumber\\
 &+ \underbrace {{}^V{k_\Gamma }{\Gamma _y}{\mathop{\rm sgn}} \left( {{}^V{v_x}} \right){}^Vv_x^2}_{{\rm{Wind - vane - like~dynamics}}}\label{eq:atau}
 % \\
% {}^V{a_{{\rm{damp}}}} = & - {}^V{k_{{\rm{damp}}}}{\mathop{\rm sgn}} \left( {{\omega _\psi }} \right)\omega _\psi ^2 \label{eq:adamp}
 \end{align} 
 where ${}^V{k_{{\rm{\tau ,x}}}}$, ${}^V{k_{{\rm{flap,x}}}}$, ${}^V{k_\Gamma }\in \mathbb R^+$, 
 % and ${}^V{k_{{\rm{damp}}}} \in \mathbb R^+$ 
 are positive constants.

% From this juncture, we present a novel concept, that is loosely non-holonomic constraints.
Due to the fact that phase~3 shown in Fig. \ref{figure:Rotate} is relatively transient, 
and considering the fact that keeping a negligible left or right translational velocity can facilitate flight stability, a non-holonomic constraint can be put forward, which is
\begin{equation}
{}^V{\dot v_y} =  {\omega _\psi }{}^V{v_x} - {}^V{k_{{\rm{drag,y}}}}{\mathop{\rm sgn}} \left( {{}^V{v_y}} \right){}^Vv_y^2 \label{eq:non-holo}
\end{equation}
% The idea behind the development of this constraint is to provide a protective envelope where the vehicle can fly safely.
Although the transient phase~3 in Fig. \ref{figure:Rotate} occurs during real flight, the intricate dynamics can be effectively managed by the trajectory controller. 
Keeping $\omega_\psi$ small, in the vertical frame dynamics, we can further use an approximation of the steady-state of \eqref{eq:non-holo} as the simplified constraint, similar to the common kinematic constraint adopt in differential mobile robot \cite{Fan-2022}:
\begin{equation}
{}^V{\dot v_y} \equiv 0,~~{}^V{v_y} \equiv 0\label{eq:non-holo1}
\end{equation}
% The term ``loosely'' is attributed to this constraint due to its partial reliance on manual enforcement, coupled with a lack of rigid confinement by its inherent dynamics.

\section{Trajectory Generation}

\subsection{Differential Flatness}
 A system is said to be differentially flat if its state variables and control inputs can be uniquely determined by a set of flat outputs and their derivatives \cite{Han-2022}.
 The choice of the flat outputs are given by
 \begin{equation}
 \bm \sigma  = {\left[ {\begin{array}{*{20}{c}}
x&y&z&\psi 
\end{array}} \right]^\top} \label{eq:sigmaa}
 \end{equation}
 where $\bm p = {\left[ {\begin{array}{*{20}{c}}
x&y&z
\end{array}} \right]^\top} \in \mathbb R^3$ is the mass center position, 
and $\psi \in \mathcal S^1$ is the yaw angle of the vertical frame. 

On the other hand, the states of system are given by ${\left[ {x,y,z,{\Gamma _x},{\Gamma _y},{\Gamma _z},\psi ,{}^V{{\dot v}_x},{}^V{{\dot v}_y},{}^V{{\dot v}_z},{\omega _x},{\omega _y},{\omega _z}} \right]^\top}$.
The following deductions are predicated on the assumption that all coefficients are ascertainable.

First of all, we have 
 \begin{align}
\bm v &= {\left[ {\begin{array}{*{20}{c}}
{\dot x}&{\dot y}&{\dot z}
\end{array}} \right]^\top}\\
{}^V{\bm v} &= R ^\top\left( \psi  \right)\bm v
 \end{align}
 such that ${}^V{\bm v}$ can be determined.
And based on \eqref{eq:vertical6}, $\omega_\psi$ can be determined.
Since the computation of ${}^V{\bm v}$ are analytical, we can compute its derivative ${}^V\dot{\bm v}$, 
such that ${}^V\dot{\bm v}$ can be determined.

It is noteworthy that the rudder induced torque is negligible in \eqref{eq:atau}, 
such that \eqref{eq:vertical7} can be rewritten as  
 \begin{align}
  \label{eq:omegapsi}
{{\dot \omega }_\psi }{\rm{ = }}{}^V{k_\Gamma }{\Gamma _y}{\mathop{\rm sgn}} \left( {{}^V{v_x}} \right){}^Vv_x^2 - {}^V{k_{{\rm{damp}}}}{\mathop{\rm sgn}} \left( {{\omega _\psi }} \right)\omega _\psi ^2
 \end{align}

Based on \eqref{eq:vertical3}, \eqref{eq:vertical5}, \eqref{eq:thrustx}, \eqref{eq:thrustz}, and \eqref{eq:omegapsi}
then $f_{{\rm{flap}}}^{\rm{2}}{\Gamma _x}$, ${\Gamma _y}$, and $f_{{\rm{flap}}}^{\rm{2}}{\Gamma _z}$ can be determined.
Then considering the fact $\Gamma _x^2 + \Gamma _y^2 + \Gamma _z^2 = 1$ and ${f_{{\rm{flap}}}} \ge 0$, and further keeping the flapping frequency not negligible, then
${\Gamma _x}$, ${\Gamma _y}$, ${\Gamma _z}$, and ${f_{{\rm{flap}}}}$ are determined.

Then the reduced attitude can be obtained through the following equation:
\begin{align}
&{{\bm q}_e} = {{{\bm q_{er}}}}/{{\left\| {{\bm q_{er}}} \right\|}},\\
&{\bm q_{er}}\triangleq s_e\left[ {\begin{array}{*{20}{c}}
{{\bm \Gamma ^{\top}}{\bm e_3} + 1}&{\bm \Gamma  \times {\bm e_3}}
\end{array}} \right],
\end{align}
where ${s_e} \in \left\{ {1, - 1} \right\}$. 
Thus we can find the attitude shown as a rotation matrix:
\begin{equation}
R = R\left( \psi  \right)R\left( {{\bm q_e}} \right)
\end{equation}
To this end, the angular velocity can be computed through the following manner:
\begin{equation}
{\left[ \bm \omega  \right]_ \times } = \dot R{R^\top}
\end{equation}
where the angular velocities $\omega_x$, $\omega_y$, and $\omega_z$  can be extracted from the skew-symmetric matrix $\left[ \bm \omega  \right]_ \times $. After taking the derivative, the angular acceleration components $\dot\omega_x$, $\dot\omega_y$, and $\dot\omega_z$ can also be determined.

Considering \eqref{eq:ConDynamics} and \eqref{eq:tau_theta},
there are two observations of $\theta_{\rm rud}$. 
Since the yaw torque is negligible, the observation of the $X$-axis can be used to determine $\theta_{\rm rud}$. 
Similarly, $\theta_{\rm ele}$ can also be determined.

In conclusion, with the output $\bm \sigma$ and its first, second, third, and forth order derivatives, both the states and the system inputs can be determined. Provided that the vertical dynamics and the constraint \eqref{eq:non-holo} holds, the FWAV dynamics and the three inputs are differentially flat. 

Lastly, let us carefully examine the equation \eqref{eq:atau} and the underlying wind-vane-like dynamics, as well as the rudder deflection effects portrayed in Fig. \ref{figure:Rotate}. Neglecting the transient rotation dynamics, this dynamics leads to a loose non-holonomic-like constraint  
\begin{equation}
\psi  = \arctan\!2\left( {\dot x,\dot y} \right) \label{eq:psigen}
\end{equation}
Since the misalignment in the transient phase can be further handled by the controller discussed later, we can exclude $\psi$
and use a 3-dimensional flat outputs, which is shown as  
\begin{equation}
 \bm \sigma  = {\left[ {\begin{array}{*{20}{c}}
x&y&z
\end{array}} \right]^\top} \label{eq:sigmaa}
 \end{equation}

\subsection{Optimization}
The utilization of polynomial trajectories is inherently suitable for highly dynamic vehicles and robots \cite{Mellinger-2011}. 
The trajectory at the $\mu$-th segment is represented as
\begin{equation}
\label{eq:poly}
{\sigma _{\mu,j}}\left( t \right) = \sum\nolimits_{i = 0}^N {{c_i}} {t^i},~~j = 1,2,3,4,~~t \in \left( {0,T} \right]
\end{equation}
where $c_i$ are the coefficients of the polynomial, $N \in \mathbb Z^+$ is the order of the polynomial, and $\sigma_{\mu,1} = x$, $\sigma_{\mu,2} = y$, and $\sigma_{\mu,3} = z$, $T\in \mathbb R^+$ is the time duration of a polynomial segment.
To this end, every polynomial starts at the time $0$,
and ends at the time $T$.
With minimal abuse of notation, 
 $\sigma _{j}$ is used instead of $\sigma _{\mu,j}$ in the context of discussing a constraint that is uniformly applied across all segments.

\subsubsection{Continuity Constraint}
When using multiple segments polynomial trajectory planning, it is important for the preceding segment of the polynomial to seamlessly connect with the subsequent segment.
The continuity constraints are maintained for each individual segment, which can be formulated as 
 \begin{align}
 {{ \sigma }_{\mu,j}}\left( T \right) &= {{ \sigma }_{\mu + 1,j}}\left( 0 \right)\\
{{\dot \sigma }_{\mu,j}}\left( T \right) &= {{\dot \sigma }_{\mu + 1,j}}\left( 0 \right)\\
{{\ddot \sigma }_{\mu,j}}\left( T \right) &= {{\ddot \sigma }_{\mu + 1,j}}\left( 0 \right)\\
{\sigma ^{\left( 3 \right)}}_{\mu,j}\left( T \right) &= {\sigma ^{\left( 3 \right)}}_{\mu + 1,j}\left( 0 \right)
 \end{align}
 where ${{\dot \sigma }_{\mu,j}}$, ${{\ddot \sigma }_{\mu,j}} \in \mathbb R$, ${\sigma ^{\left(3 \right)}}_{\mu,j}$, with $j =1,2,3,4$, are the first, second, and third order of a specific entry of the trajectory, respectively.

\subsubsection{Boundary Constraint}
The incorporation of boundary constraints in trajectory generation tasks serves a crucial purpose within the context of ensuring the behaviors at the initiation and termination points. For example, with respect to the initial and terminal condition of the $x$-direction dynamics in the inertia frame, we have
 \begin{align}
&{\sigma _{1,1}}\left( 0 \right) = {x_s},~{{\dot \sigma }_{1,1}}\left( 0 \right) = {{\dot x}_s},~{{\ddot \sigma }_{1,1}}\left( 0 \right) = {{\ddot x}_s}\\
&{\sigma _{M,1}}\left( T \right) = {x_t},~{{\dot \sigma }_{M,1}}\left( T \right) = {{\dot x}_t},~{{\ddot \sigma }_{M,1}}\left( T \right) = {{\ddot x}_t}
 \end{align}
 where ${x_s}$, ${{\dot x}_s}$, and ${{\ddot x}_s} \in \mathbb R$ are the position, velocity, and acceleration boundary constraints at the initiation, respectively. Simultaneously, ${x_t}$, ${{\dot x}_t}$, and ${{\ddot x}_t} \in \mathbb R$ correspond to the termination $x$-directional states. And $M\in\mathbb Z^+$ is the number of the polynomial segments. These constraints can be applied to the dynamics along both $Y$-axis and $Z$-axis of the vertical frame.

\subsubsection{Kinodynamic Constraint}
When generating trajectories for the FWAV, it is important to consider not only its kinematic feasibility but also its dynamic feasibility. The horizontal velocity constraint of the FWAV is firstly considered:
 \begin{equation}
\sqrt {\dot \sigma _1^2\left( t \right) + \dot \sigma _2^2\left( t \right)}  \le {v_{{\rm{h,max}}}} \label{eq:vh-con}
 \end{equation}
 where $ {v_{{\rm{h,max}}}}\in \mathbb R^+$ is the maximum horizontal velocity.
The vertical velocity constraint is then considered:
 \begin{equation}
 \left\| {{{\dot \sigma }_3}\left( t \right)} \right\| \le {v_{{\rm{v,max}}}} \label{eq:vv-con}
 \end{equation}
 where $v_{{\rm{v,max}}}\in \mathbb R^+$ is the maximum vertical velocity.

%  Another important kinodynamics constraint is the azimuth angular velocity limitation, based on the constraint \eqref{eq:psigen}, we have
% \begin{equation}
% \arctan\!2\left( {{{\dot \sigma }_1},{{\dot \sigma }_2}} \right) \le {{\dot \psi }_{\max }}
% \end{equation}
%  where ${{\dot \psi }_{\max }}\in \mathbb R^+$ is the maximum azimuth angular velocity.
%  In order to keep the stable flight of the FWAV, 
%  the heading direction constraint should also be considered:
% \begin{equation}
% 1 - \cos {\theta _v}\cos \sigma_4  - \sin {\theta _v}\sin \sigma_4 = 0 \label{eq:hd-con}
% \end{equation}
% % where $\delta {\psi _{\max }}$ is the maximum heading direction deviation, 
% and the moving direction $\theta _v$ is defined as 
% \begin{equation}
% {\theta _v} = \arctan\!2\left( {{{\dot \sigma }_1},{{\dot \sigma }_2}} \right)
% \end{equation}
% where the function  $\arctan\!2(\star_1,\star_2) $ is a mathematical function that calculates the angle in radians between the positive $x$-axis and the point $(\star_1,\star_2)$.
% Due to this constraint, the development freedom of $\psi$ is exhausted, such that $\psi$ is completely determined by the derivative ${\dot \sigma }_1$ and ${\dot \sigma }_2$.

% It is worth noting that, if the non-holonomic constraint in last section is considered to be stringent, it can also be interpreted as a kinodynamic constraint.

\subsubsection{Obstacle Constraint}
Incorporating obstacle constraints is essential during trajectory generation as it ensures safe navigation by preventing collisions with obstacles in real-world environments, which is shown as
\begin{equation}
{\left[ {\begin{array}{*{20}{c}}
{{\sigma _1}}&{{\sigma _2}}&{{\sigma _3}}
\end{array}} \right]^{\top}} \in {D_{{\rm{ob}}}}
\end{equation}
where the obstacle domain ${D_{{\rm{ob}}}} \in \mathbb R^3$ can be formulated as a finite union of spheres \cite{Wang-2020} or polyhedrons \cite{Wolff-2014}.

\subsubsection{Objective}
Minimum snap trajectories ensure smooth and natural motion profiles for FWAVs by minimizing abrupt changes in acceleration, jerk, and higher-order derivatives.
The minimum snap objective is formulated as 
\begin{equation}
\label{eq:minsnap}
\min \int_0^{M \cdot T} {\left( {{\mu _p}{{\sum\nolimits_{j = 1}^3 {\left\| {\frac{{{{\rm{d}}^4}{\sigma _j}}}{{{\rm{d}}{{\rm{t}}^4}}}} \right\|} }^2} } \right)} {\rm{dt}}
\end{equation}
where $\mu_p\in \mathbb R^+$ is a positive constant, which is chosen to make the integrand
non-dimensional.
The constraints mentioned above can be exerted on this problem according to actual situations.
Due to the fact that the rudder and the elevator deflection angle, $\theta_{\rm rud}$ and $\theta_{\rm ele}$
relate to the fourth order derivative of the trajectory.
By setting a minimum snap objective, it is evident that lower control effort can be achieved in a straightforward manner. 
An additional significant advantage of incorporating a minimum snap objective is the preservation of stable aerodynamic conditions, which holds paramount importance in flapping wing flights.

\section{Trajectory Tracking}
Trajectory tracking control serves as a crucial means for achieving the desired trajectory generated for FWAVs. 
% By employing a robust control approach with prescribed performance, it is possible to effectively maintain the FWAV within its designated flight envelope, where the dynamics is governed by the vertical frame dynamics. 
% This not only ensures stable and reliable flight but also enables enhanced maneuverability and control capabilities.
Simultaneously, implementing the time-scale separation between the position dynamics and the attitude dynamics \cite{Pre1, Paranjape-2013}, we can 
develop and analyze the control strategies in the vertical frame.

\subsection{Controller Development}
According to Fig. \ref{figure:Rotate}, phase 2 is the only unsustainable mode, which is evanescent in most maneuvering. And from other phases, we can see that there is relatively fixed relationship between the rudder deflection angle and the reduced attitude component $\Gamma_y$, when the transient behaviors are neglected. To this end, the vertical frame yaw actuating torque \eqref{eq:atau} can be rewritten as 
\begin{align}
{}^V{a_{{\rm{\tau ,z}}}} = -{  \left( {{}^V{\bar k_{\Gamma}}{\mathop{\rm sgn}} \left( {{}^V{v_z}} \right){}^Vv_z^2 + {}^V{\bar k_{{\rm{flap,x}}}}f_{{\rm{flap}}}^2} {\Gamma _z}\right){\Gamma _{y}}}\label{eq:atau1}
 \end{align}
 where ${}^V{\bar k_{\Gamma}}$, ${}^V{\bar k_{{\rm{flap,x}}}} \in \mathbb R^+$ are positive constants. 

 Therefore, the vertical frame dynamics can then be seen as the outer-loop dynamics, where the reduced attitude $\Gamma$ is considered as the control input. 

 The position and velocity errors can both be defined in the inertia frame, which are given by
\begin{align}
{\bm e_p} &= {\bm p_d} - \bm p\label{eq:ep}\\
{\bm e_v} &= {\bm v_d} - {\bm v}\label{eq:ev}
 \end{align}
 where $\bm e_p, \bm e_v \in \mathbb R^3$, and $\bm p_d, \bm v_d \in \mathbb R^3$ are the desired 
 translational position and the velocity. The generated reference trajectory is denoted as $\bm \sigma_r, \dot{\bm \sigma}_r \in \mathbb R^3$, which are straightforwardly ${\left[ {\begin{array}{*{20}{c}}
{{\sigma _1}}&{{\sigma _2}}&{{\sigma _3}}
\end{array}} \right]^{\top}}$ and ${\left[ {\begin{array}{*{20}{c}}
{{\dot\sigma _1}}&{{\dot\sigma _2}}&{{\dot\sigma _3}}
\end{array}} \right]^{\top}}$. In addition $\bm p_d = \bm \sigma_r$, however, $\bm v_d \ne \dot{\bm \sigma}_r$, because  $\bm v_d$ should also consider the feedback errors.
The azimuth angle error is given by
\begin{equation}
{e_\psi } = \sqrt 2  - \sqrt {1 + \cos \left( {{\psi _d} - \psi } \right)}  \label{eq:epsi}
\end{equation}
where $ e_\psi \in \mathbb R^{+}\cup\{0\}$, and $ e_\psi = 0$ if and only if $\psi  = {\psi _d}$.
However, $\sigma_4$ cannot be used as $\psi _d$, but prior feedforward information at best, where feedback observation should also be considered.
Correspondingly, the angular velocity error is given by 
\begin{equation}
{e_{\omega \psi} } = {\omega _{\psi d}} - \omega_\psi \label{eq:omega}
\end{equation}
where ${e_{\omega \psi} }, {\omega _{\psi d}} \in \mathbb R$ are the azimuth angular velocity error, and the desired azimuth angular velocity, respectively.
And the equation ${{\dot \psi }_d} = {\omega _{\psi d}}$ does not necessarily holds, such that $\omega _{\psi d}$ can be freely designed.

With these errors defined, naturally, the following two Lyapunov candidate functions can be proposed:
\begin{align}
{V_{1}} = &\frac{1}{2}\bm e_p^\top{K_p^{-1}}{\bm e_p} + \frac{1}{2}\bm e_v^\top{K_v^{-1}}{\bm e_v}\\
 {V_{2}} =& {k_\psi^{-1} }\left( \sqrt 2  - \sqrt {1 + {\rm c} \left( {{\psi _d} - \psi } \right)}  \right) + \frac{1}{2}{k_\omega ^{-1}} e_{\omega \psi}^2
\end{align}
where $K_p$, $K_v$, and $K_\omega \in \mathbb R^{3 \times 3}$ are positive definite, diagonal matrices,
$k_\psi \in \mathbb R^+$ is a positive constant. 
Upon closer inspection, we can consider the positional subsystem formulated by \eqref{eq:vertical1}-\eqref{eq:vertical5} 
and the heading subsystem formulated by \eqref{eq:vertical6}-\eqref{eq:vertical7} can be viewed as the perturbed and the perturbing system of a cascade nonlinear system, which correspond to $V_{1}$ and $V_{2}$, respectively \cite{Panteley-1998}.
In the subsequent equations, ${\rm c}\psi$ and ${\rm s}\psi$ are the abbreviations for $\cos\psi$ and $\sin\psi$, respectively.

Consider the derivative of the first Lyapunov candidate $V_{1}$, there exists
\begin{align}
{{\dot V}_{1}} = &{\bm e}_p^\top K_p^{ - 1}{{\dot {\bm e}}_p} + {\bm e}_v^\top K_v^{ - 1}{{\dot {\bm e}}_v}\nonumber\\
 = &{\bm e}_p^\top K_p^{ - 1}\left( {{{\dot {\bm \sigma} }_r} - {\bm v}} \right) + {\left( {{{\bm v}_d} - {\bm v}} \right)^\top }K_v^{ - 1}\left( {{{\dot {\bm v}}_d} - \dot {\bm v}} \right)\nonumber\\
 = &{\bm e}_p^\top K_p^{ - 1}\left( {{{\dot {\bm \sigma} }_r} - {{\bm v}_d}} \right)\nonumber\\
 &+ {\left( {{{\bm v}_d} - {\bm v}} \right)^\top }\left[ {K_v^{ - 1}\left( {{{\dot {\bm v}}_d} - \dot {\bm v}} \right) + K_p^{ - 1}{{\bm e}_p}} \right]\label{eq:V11}
\end{align}
where $\bm v_d$ is designed as  
\begin{equation}
{\bm v_d} = {\dot {\bm \sigma} _r} + {K_p}\bm\tanh({\bm e_p} )\label{eq:vd}
\end{equation}
and we further expect that $\dot{\bm v}$ satisfies 
\begin{equation}
\dot {\bm v} = {{\dot {\bm v}}_d} + {K_v}K_p^{ - 1}{\bm\tanh({\bm e_p})} + {K_v}{\bm\tanh({\bm e_v})} \label{eq:dotv}
\end{equation}
and the right side is denoted as $\bm a_d \in \mathbb R^3$.
The saturation function $ {\mathbf {tanh}}({\bm \star}): \mathbb{R}^3 \to \mathbb{R}^3$ is a vector-valued function that applies the hyperbolic tangent function operation ${\rm tanh}(\star) $ to each element $\star$ of the input vector $\bm \star$, and subsequently combines the individual results.

Then the following desired values can be obtained:
\begin{align}
{\psi _d} &= \arctan 2\left( {{a_{xd}},{a_{yd}}} \right) \label{eq:psidd}\\
{}^V{{\dot v}_{xd}} &= \sqrt {a_{xd}^2 + a_{yd}^2} \label{eq:dvxd}\\
{}^V{{\dot v}_{zd}} &= {a_{zd}} \label{eq:dvzd}
\end{align}
where ${\bm a_d} = {\left[ {\begin{array}{*{20}{c}}
{{a_{xd}}}&{{a_{yd}}}&{{a_{zd}}}
\end{array}} \right]^\top}$.
The azimuth angle $\psi_d$ is provided in \eqref{eq:psidd}, and not shown as an explicit function of $\sigma_4$.

On the one hand, the $x$-directional component of  the reduced attitude $\Gamma_x$, the $z$-directional component  $\Gamma_z$,
and the flapping wing frequency $f_{\rm flap}$ are regarded as inputs of the positional subsystem in the 
vertical frame dynamics. 
On the other hand, the $y$-directional component of the reduced attitude $\Gamma_y$ and $\omega_\psi$ belong to the higher order heading subsystem.
Thus, \eqref{eq:vertical3} and \eqref{eq:vertical5} can be rewritten as 
\begin{align}
&{}^V{{\dot v}_x} = -{{{k_{{\rm{tf}}}}f_{{\rm{flap}}}^2{\Gamma _x}} \mathord{\left/
 {\vphantom {{{k_{{\rm{tf}}}}f_{{\rm{flap}}}^2{\Gamma _x}} m}} \right.
 \kern-\nulldelimiterspace} m} - {{{}^V{k_{{\rm{d,x}}}}{\mathop{\rm sgn}} \left( {{}^V{v_x}} \right){}^Vv_x^2} \mathord{\left/
 {\vphantom {{{}^V{k_{{\rm{d,x}}}}{\mathop{\rm sgn}} \left( {{}^V{v_x}} \right){}^Vv_x^2} m}} \right.
 \kern-\nulldelimiterspace} m}\\
&{}^V{{\dot v}_z} = {{{k_{{\rm{tf}}}}f_{{\rm{flap}}}^2{\Gamma _z}} \mathord{\left/
 {\vphantom {{{k_{{\rm{tf}}}}f_{{\rm{flap}}}^2{\Gamma _z}} m}} \right.
 \kern-\nulldelimiterspace} m} - {{{}^V{k_{{\rm{d,z}}}}{\mathop{\rm sgn}} \left( {{}^V{v_z}} \right){}^Vv_z^2} \mathord{\left/
 {\vphantom {{{}^V{k_{{\rm{d,z}}}}{\mathop{\rm sgn}} \left( {{}^V{v_z}} \right){}^Vv_z^2} m}} \right.
 \kern-\nulldelimiterspace} m} - g
\end{align}
When composing the desired reduced attitude, the component $\Gamma_y$ is considered to be $0$, such that $\Gamma _x^2 + \Gamma _z^2 = 1$.
Further considering the desire provided in \eqref{eq:dvxd} and \eqref{eq:dvzd}, nominal inputs of the positional subsystem can be consequently computed as
\begin{align}
f_{\rm flap}^2 &= k_{\rm tf}^{ - 1}m\sqrt {\dot v_{cx}^2 + \dot v_{cz}^2} \\
-{\Gamma _{xd}} &= {{{{\dot v}_{cx}}} \mathord{\left/
 {\vphantom {{{{\dot v}_{cx}}} {\sqrt {\dot v_{cx}^2 + \dot v_{cz}^2} }}} \right.
 \kern-\nulldelimiterspace} {\sqrt {\dot v_{cx}^2 + \dot v_{cz}^2} }}\\
{\Gamma _{zd}} &= {{{{\dot v}_{cz}}} \mathord{\left/
 {\vphantom {{{{\dot v}_{cz}}} {\sqrt {\dot v_{cx}^2 + \dot v_{cz}^2} }}} \right.
 \kern-\nulldelimiterspace} {\sqrt {\dot v_{cx}^2 + \dot v_{cz}^2} }}
\end{align} 
where the combined velocity changing rates ${{{\dot v}_{cx}}}$ and ${{{\dot v}_{cz}}} \in \mathbb R$ are given by
\begin{align}
{{\dot v}_{cx}} &= {}^V{{\dot v}_{xd}} + {{{}^V{k_{{\rm{d,x}}}}{\mathop{\rm sgn}} \left( {{}^V{v_x}} \right){}^Vv_x^2} \mathord{\left/
 {\vphantom {{{}^V{k_{{\rm{d,x}}}}{\mathop{\rm sgn}} \left( {{}^V{v_x}} \right){}^Vv_x^2} m}} \right.
 \kern-\nulldelimiterspace} m} \label{eq:vcx}\\
{{\dot v}_{cz}} &= {}^V{{\dot v}_{zd}} + {{{}^V{k_{{\rm{d,z}}}}{\mathop{\rm sgn}} \left( {{}^V{v_z}} \right){}^Vv_z^2} \mathord{\left/
 {\vphantom {{{}^V{k_{{\rm{d,z}}}}{\mathop{\rm sgn}} \left( {{}^V{v_z}} \right){}^Vv_z^2} m}} \right.
 \kern-\nulldelimiterspace} m} + g \label{eq:vcz}
\end{align}

% Then we consider the derivative of the second Lyapunov candidate $V_{2}$, the following equation are obtained:
% \begin{align}
% {{\dot V}_{2}} = &k_\psi ^{ - 1} \left( {s{\psi _d}c\psi  - c{\psi _d}s\psi } \right){{\dot \psi }_d} + k_\psi ^{ - 1}\left( {c{\psi _d}s\psi  - s{\psi _d}c\psi } \right)\dot \psi \nonumber\\
%  & + k_\omega ^{ - 1}{e_{\omega \psi }}\left( {{{\dot \omega }_{\psi d}} - {{\dot \omega }_\psi }} \right) \label{eq:dV12}
% \end{align}

% By considering the fact that
% \begin{equation}
% s{\psi _d}c\psi  - c{\psi _d}s\psi  = \sin \left( {{\psi _d} - \psi } \right)
% \end{equation}
% the derivative \eqref{eq:dV12} is shown as 
% \begin{align}
% {{\dot V}_{2}} =& k_\psi ^{ - 1}{\rm s} \left( {{\psi _d} - \psi } \right)\left( {{{\dot \psi }_d} - \dot \psi } \right) + k_\omega ^{ - 1}{e_{\omega \psi }}\left( {{{\dot \omega }_{\psi d}} - {{\dot \omega }_\psi }} \right)\nonumber\\
%  =& k_\psi ^{ - 1}{\rm s} \left( {{\psi _d} - \psi } \right)\left( {{{\dot \psi }_d} - {\omega _{\psi d}}} \right)\nonumber\\
%  &+ k_\omega ^{ - 1}\left( {{\omega _{\psi d}} - {\omega _\psi }} \right)\left( {{k_\omega }k_\psi ^{ - 1}{\rm s} \left( {{\psi _d} - \psi } \right) + {{\dot \omega }_{\psi d}} - {{\dot \omega }_\psi }} \right) \label{eq:dV2}
% \end{align}
In order to avoid the possible chattering at the state where ${\rm{s}}\left( {{\psi _d} - \psi } \right) = 0$ and ${\rm{c}}\left( {{\psi _d} - \psi } \right) =  - 1$, a hysteric term is introduced.
Thus, the Lyapunov candidate can be modified to the following form:
\begin{align}
{V_2} = &k_\psi ^{ - 1}\left( \sqrt 2  - {h_\psi }{\mathop{\rm sgn}} \left( {{\rm{s}}\left( {{\delta _\psi }} \right)} \right)\sqrt {1 + {\mathop{\rm c}\nolimits} \left( {{\delta _\psi }} \right)}  \right) \nonumber\\
&+ \frac{1}{2}k_\omega ^{ - 1}e_{\omega \psi }^2 \label{eq:newV2}
\end{align}
where we define the deviation $\delta _\psi = \psi_d - \psi$.
Accordingly, the corresponding derivative is given by
\begin{align}
{{\dot V}_{2}} =& \frac{{k_\psi ^{ - 1}}}{2}\frac{{{h_\psi }\left|{\rm{s}}\left( {{\delta _\psi }} \right)\right|}}{{\sqrt {1 + {\rm{c}}\left( {{\delta _\psi }} \right)} }}\left( {{{\dot \psi }_d} - {\omega _{\psi d}}} \right)\nonumber\\
 &+ k_\omega ^{ - 1}\left( {{\omega _{\psi d}} - {\omega _\psi }} \right)\left( {\frac{{k_\psi ^{ - 1}{h_\psi }\left|{\rm{s}}\left( {{\delta _\psi }} \right)\right|}}{{2\sqrt {1 + {\rm{c}}\left( {{\delta _\psi }} \right)} }} + {{\dot \omega }_{\psi d}} - {{\dot \omega }_\psi }} \right) \label{eq:newdV2}
\end{align}

Therefore, the control strategy ${\omega _{\psi d}}$ is designed as
\begin{equation}
{\omega _{\psi d}} = {{\dot \psi }_d} + k_\psi{h_\psi }{\sqrt {1 - {\rm{c}}\left( {{\delta _\psi }} \right)} }\label{eq:omegapsid}
\end{equation}
where the term ${h_\psi }\sqrt {1 - {\rm{c}}\left( {{\delta _\psi }} \right)}$ is 
used instead of ${h_\psi }{{\left| {{\rm{s}}\left( {{\delta _\psi }} \right)} \right|} \mathord{\left/
 {\vphantom {{\left| {{\rm{s}}\left( {{\delta _\psi }} \right)} \right|} {\sqrt {1 + {\rm{c}}\left( {{\delta _\psi }} \right)} }}} \right.
 \kern-\nulldelimiterspace} {\sqrt {1 + {\rm{c}}\left( {{\delta _\psi }} \right)} }}$ to circumvent the removable singularities therein at $\delta _\psi = \pi + 2 k \pi, k \in \mathbb Z$.

With $ \delta \in \mathbb R^+$ being a small positive constant, the hysteretic term $h_\psi$ dynamics can be formulated as
\begin{align}
\label {eq:HHlaw2}
&{h_\psi^ + } \in {\mathop{\rm \overline {{\mathop{\rm sgn}} } }} \left( {{{\mathop{\rm s}\nolimits} \left( {\delta _\psi } \right)}} \right),
\nonumber\\&~~~~\left(h_\psi{{\mathop{\rm s}\nolimits} \left( {\delta _\psi } \right)} \le  - \delta,~{\rm and}~{{\mathop{\rm c}\nolimits} \left( {\delta _\psi } \right)} \le  0 \right),~ {\rm or}~{{\mathop{\rm c}\nolimits} \left( {\delta _\psi } \right)} >  0 \\
\label {eq:HHlaw1}
&\dot h_\psi = 0,~~~~~~~~~~~~~~~~~~~~~~~~~~~~~~~~~~~~~~~~~{\rm otherwise}, 
\end{align}
where $h_\psi^ +$ represents the updated logic variable, and the set-valued function $\overline {{\mathop{\rm sgn}} } \left( \star \right)$ is defined as
\begin{align}
\label {eq:Setv}
\overline {{\mathop{\rm sgn}} } \left( \star \right) = \left\{ {\begin{array}{*{20}{c}}
   {{\mathop{\rm sgn}} \left( \star \right),\left| \star \right| > 0}  \\
   {\left\{ { - 1,1} \right\},\star = 0.}  \\
\end{array}} \right.
\end{align}
Suppose that the vehicle dynamics is constrained, such that the uncertain control gain shown in equation \eqref{eq:atau1} satisfies  
\begin{equation}
 {l_{\Gamma ,\min }} \le {{}^V{\bar k_{\Gamma}}{\mathop{\rm sgn}} \left( {{}^V{v_z}} \right){}^Vv_z^2 + {}^V{\bar k_{{\rm{flap,x}}}}f_{{\rm{flap}}}^2} {\Gamma _z} \le {l_{\Gamma ,\max }}
\end{equation}
which can be achieved by exerting constraints on both the velocity and input, such as the upper and the lower bounds ${l_{\Gamma ,\min }} $ and ${l_{\Gamma ,\max}} \in \mathbb R^+$, 
however, for the sake of conciseness, it is not addressed within the scope of this paper.
And we further develop the desire of the reduced attitude component $\Gamma_{yd}$ as 
\begin{align}
{\Gamma _{yd}} = &-(\frac{{{k_\omega }}}{{{l_{\Gamma ,\min}}}} - \frac{{{k_\omega }}}{{{l_{\Gamma ,\max}}}}) {\mathop{\rm sgn}} \left( {{\omega _{\psi d}} - {\omega _\psi }} \right) \nonumber\\& \cdot\left| {\frac{1}{2}k_\psi ^{ - 1}{h_\psi }\sqrt {1 - {\rm{c}}\left( {{\delta _\psi }} \right)}  + {{\dot \omega }_{\psi d}}} \right|\nonumber\\
 & - \frac{{{k_\omega }}}{{{l_{\Gamma ,\max }}}}\left( {\frac{1}{2}k_\psi ^{ - 1}{h_\psi }\sqrt {1 - {\rm{c}}\left( {{\delta _\psi }} \right)} + {{\dot \omega }_{\psi d}}} \right) \nonumber\\
 &- {k_\omega }\left( {{\omega _{\psi d}} - {\omega _\psi }} \right) \label{eq:gammay}
\end{align}
One can rely on the robust term  ${\mathop{\rm sgn}} \left( {{\omega _{\psi d}} - {\omega _\psi }} \right)$ to circumvent the exact feedforward compensation, however, $\frac{{{k_\omega }}}{{{l_{\Gamma ,\max }}}}\left( {{k_\omega }k_\psi ^{ - 1}{\rm{s}}\left( {\delta _\psi } \right)+ {{\dot \omega }_{\psi d}}} \right)$ is still invoked to alleviate the burden of the discontinuous term, such that the controller can avoid extremely high magnitude chattering of $\Gamma_{yd}$.

\subsection{Stability Analysis}
In this subsection, we verify the asymptotic stability of the positional subsystem \eqref{eq:vertical1}-\eqref{eq:vertical5}, and the stability of the heading subsystem \eqref{eq:vertical6}-\eqref{eq:vertical7},
then use the cascade system stability theory to analyze the overall system behavior.

\begin{Proposition}
Consider the positional subsystem \eqref{eq:vertical1}-\eqref{eq:vertical5}, equipped with the control strategy presented as \eqref{eq:vd} and \eqref{eq:dotv}, where the expectation \eqref{eq:dotv} is immediately satisfied,
that is, the input of the perturbing system, heading subsystem, is zero. Then the closed-loop positional subsystem is globally uniformly asymptotically stable, whose equilibrium point is $\bm e_p = \left[ {\begin{array}{*{20}{c}}0&0&0\end{array}} \right]^\top$ and $\bm e_v = \left[ {\begin{array}{*{20}{c}}0&0&0\end{array}} \right]^\top$. \label{Pro:positional1}
\end{Proposition}

\begin{proof}
Let us consider the behavior of the Lyapunov candidate $V_{1}$.
It is obvious that $V_{1} = 0$, if and only if $\bm e_p = \left[ {\begin{array}{*{20}{c}}0&0&0\end{array}} \right]^\top$ and $\bm e_v = \left[ {\begin{array}{*{20}{c}}0&0&0\end{array}} \right]^\top$.
And there exists a class $\mathcal K$ function $\alpha(\bm \star)$, such that $\alpha(\bm e_p, \bm e_v) \le V_{1}$, for example $\lambda_{1} V_{1}$ with $0 < \lambda_{1} < 1$.
Substituting \eqref{eq:vd} and \eqref{eq:dotv} into \eqref{eq:V11}, the derivative of \eqref{eq:V11} becomes 
% \section{Simulations }
\begin{equation}
{\dot V_{1}} =  - {\bm e}_p^\top{\bm \tanh} \left( {{{\bm e}_p}} \right) - {\bm e}_v^\top{\bm \tanh} \left( {{{\bm e}_v}} \right) \label{eq:dV11new}
\end{equation}
where ${\dot V_{1}} < 0$ for $\forall {\bm e_v},{\bm e_p} \in \mathbb R^3 $ except the origin, $\bm e_p = \left[ {\begin{array}{*{20}{c}}0&0&0\end{array}} \right]^\top$ and $\bm e_v = \left[ {\begin{array}{*{20}{c}}0&0&0\end{array}} \right]^\top$.
On the other hand, $V_{1}$ is continuously differentiable, and $\dot V_{1}$ is radially infinite. Based on these observations, it can be concluded that, the positional subsystem is globally uniformly asymptotically stable \cite{Haddad-2018}, and it follows that ${\bm e_p} \to \left[ {\begin{array}{*{20}{c}}0&0&0\end{array}} \right]^\top$, 
and ${\bm e_v} \to \left[ {\begin{array}{*{20}{c}}0&0&0\end{array}} \right]^\top$, as $t \to \infty$.
\end{proof}

Then the behavior of $V_{2}$ is investigated. The discrete nature of the heading subsystem is injected
by both the hysteretic term $h_\psi$ and the robust term ${\mathop{\rm sgn}} \left( {{\omega _{\psi d}} - {\omega _\psi }} \right)$. The discontinuous dynamics elicited by the robust term is considered to be the flows, because the discontinuity emerges in the derivative of system states, thus takes effect only when time evolves. On the other hand, the jumps elicited by the hysteretic term is rendered as the jumps of the hybrid system.
Please notice that the changes of $h_\psi$ when ${{\mathop{\rm c}\nolimits} \left( {\delta _\psi } \right)} >  0$ does not invoke any jumps, since ${\sqrt {1 - {\rm{c}}\left( {{\delta _\psi }} \right)} } = 0$ there.

\begin{Proposition}
Consider the heading subsystem \eqref{eq:vertical6}-\eqref{eq:vertical7}, equipped with the control strategy presented as \eqref{eq:omegapsid} and \eqref{eq:gammay}.
If the absolute value of the derivative of the heading direction $ \left| {{{\dot \psi }_d}} \right|$ is always bounded, then the equilibrium point of the closed loop hybrid heading subsystem is globally practically asymptotically stable. 
And the equilibrium points is shown as $\delta_\psi = 0$ and $\omega = \omega_d$.
\label{Pro:positional2}
\end{Proposition}

\begin{proof}
Let us first check the flow behavior of the hybrid heading subsystem. By substituting \eqref{eq:omegapsid} and \eqref{eq:gammay} into \eqref{eq:newdV2}, we have
\begin{equation}
{\dot V_2} \le  - \frac{1}{2}{\left( {1 - {\rm{c}}\left( {{\delta _\psi }} \right)} \right)^2} - {\left( {{\omega _{\psi d}} - {\omega _\psi }} \right)^2}\label{eq:dV2flow}
\end{equation}
Based on this observation, it follows that ${\dot V_{2}} < 0$ for $\forall \delta_\psi \in \mathcal S^1$, ${e_{\omega\psi}} \in \mathbb R$ except the condition that $\psi_d = \psi$ and $\omega _{\psi d} = \omega _{\psi}$.

Then the jump behavior is checked. The necessary conditions for the jump behavior are $\left| {{\rm{s}}\left( {{\delta _\psi }} \right)} \right| = \delta$ and $ {{\rm{c}}\left( {{\delta _\psi }} \right)}  < 0$. The Lyapunov candidate change is 
\begin{align}
V_2^ +  - {V_2} = &2k_\omega ^{ - 1}{k_\psi }\sqrt {1 + \sqrt {1 - {\delta ^2}} }  \cdot \left| {{{\dot \psi }_d}} \right| \cdot \left| {{\omega _\psi }} \right|\nonumber\\
 &- 2k_\psi ^{ - 1}\sqrt {1 - \sqrt {1 - {\delta ^2}} } \label{eq:jumpchageV2}
\end{align}
Considering the saturation nature of \eqref{eq:vd} and \eqref{eq:dotv}, and the kinodynamic constraints 
that can be exerted on the reference trajectory, it is known that $\left| {{{\dot \psi }_d}} \right|$ can be seen as bounded, whose upper bound is then denoted as $\dot\psi_{d,\max} \in \mathbb R^+$. Then, as long as the heading direction angular velocity $\omega_\psi$ satisfies the following condition:
\begin{align}
\left| {{\omega _\psi }} \right| < \frac{{{k_\omega }k_\psi ^{ - 2}\sqrt {1 - \sqrt {1 - {\delta ^2}} } }}{{{{\dot \psi }_{d,\max }}\sqrt {1 + \sqrt {1 - {\delta ^2}} } }}\label{eq:omegapsicon}
\end{align}
whose right hand side is denoted as $\overline \omega_\psi$,
the decreasing of $V_2$ holds during each jump, such that
\begin{align}
V_2^ +  - {V_2} < 0 \label{eq:V2jumpde}
\end{align}
Due to the strictly monotonic decreasing  of $V_2$ along both flows and jumps, by considering the extreme condition of the the system initial condition, there exists
\begin{align}
&{V_2}\left( 0 \right) = \frac{{k_\omega ^{ - 1}}}{2}\bar \omega _\psi ^2 \Rightarrow \nonumber\\
&\frac{{k_\omega ^{ - 1}}}{2}\omega _{\psi ,\max }^2 + \frac{{k_\omega ^{ - 1}}}{2}\dot \psi _{d,\max }^2 + \sqrt 2 k_\psi ^{ - 1} = \frac{{k_\omega ^{ - 1}}}{2}\bar \omega _\psi ^2 \Rightarrow \nonumber\\
&{\omega _{\psi ,\max }} = \sqrt {\bar \omega _\psi ^2 - \dot \psi _{d,\max }^2 - 2\sqrt 2 k_\psi ^{ - 1}{k_\omega }} 
\end{align}
And if $\left|\omega_\psi(0)\right| < \omega_{\psi,\max}$, then \eqref{eq:omegapsicon} always holds after initialization, which guarantees the decreasing of $V_2$ during jumps. 
It can be found that, by taking sufficiently large $k_\omega$ and $k_\psi$ being positive and sufficiently close to $0$, arbitrarily large value of $\omega_{\psi,\max}$ can be obtained.

Moreover, the time interval during which the variable ${\rm{s}}\left( {{\delta _\psi }} \right)$ changes from $- \delta$ to $\delta$ cannot be infinitesimal, such that the solution is complete and its time
domain is unbounded in the ordinary time direction \cite{Goebel-2009}. Then based on the \emph{Theorem 20} in \cite{Goebel-2009}, we know that the system is now uniformly asymptotically stable with respect to $\left|\omega_\psi(0)\right| < \omega_{\psi,\max}$.
\end{proof}

\begin{Proposition}
 Suppose the following two conditions are satisfied:

 (a) the closed-loop positional subsystem is globally uniformly asymptotically stable and the only equilibrium point is shown as $\bm e_p = \left[ {\begin{array}{*{20}{c}}0&0&0\end{array}} \right]^\top$ and $\bm e_v = \left[ {\begin{array}{*{20}{c}}0&0&0\end{array}} \right]^\top$,

 (b) the closed-loop heading subsystem is globally practically asymptotically stable.

 then when the desired heading direction satisfies that $\left|\dot\psi_{d}\right| \le \dot\psi_{d,\max}$, with dedicatedly chosen controller parameters $k_\omega$ and $k_\psi$, the closed loop cascaded system of \eqref{eq:vertical1}-\eqref{eq:vertical7} is asymptotically stable with respect to a closed set $\mathbb R^3 \times \mathbb R^3 \times \mathcal S^1 \times [-\overline e_{\omega\psi}, \overline e_{\omega\psi}]$, where $\overline e_{\omega\psi} = \omega_{\psi,\max} - \dot\psi_{d,\max} > 0$. And the equilibrium point is shown as $\bm e_p = \left[ {\begin{array}{*{20}{c}}0&0&0\end{array}} \right]^\top$, $\bm e_v = \left[ {\begin{array}{*{20}{c}}0&0&0\end{array}} \right]^\top$, $\delta_\psi = 0$ and $\omega = \omega_d$.
\end{Proposition} 

\begin{proof}
By virtue of the global practical asymptotic stability in (b), the closed-loop heading subsystem is asymptotically stable to the closed set $\mathcal S^1 \times [-\overline e_{\omega\psi}, \overline e_{\omega\psi}]$.
It now follows from \emph{Proposition 3.1} in \cite{Maggiore-2019}, and the asymptotic stability can be guaranteed.
\end{proof}

\begin{figure*}[t]
      \centering
      \includegraphics[width=5.5in]{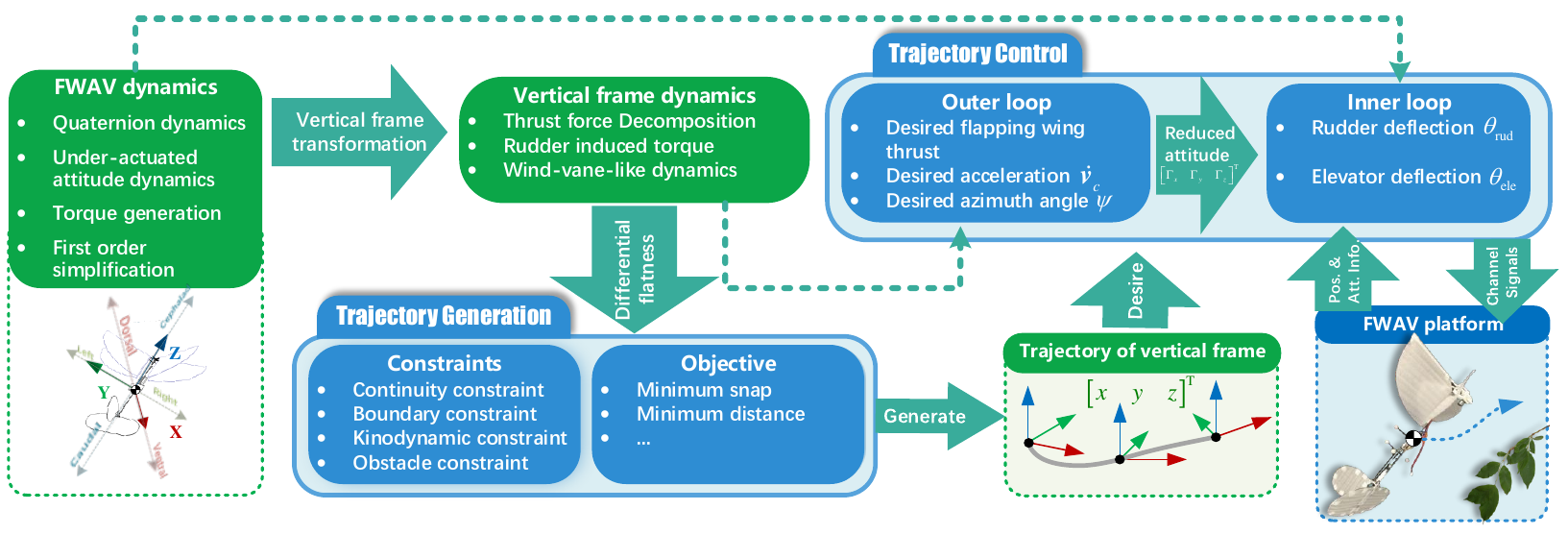}\\
      \caption{Schematic for the proposed trajectory generation strategy, tracking control strategy, and their subsequent implementation in the autonomous flight of the developed FWAV.
      }
      \label{figure:overall}
\end{figure*}

\subsection{Practical Adaptation}
So far, we have provided the outer-loop control strategy for the trajectory tracking task. 
However,  there is still a requirement for several adaptations aimed at practical application.

First of all, since the velocity in $z$-direction is relatively limited, and the windward area is also 
relatively small, therefore, we neglect the resistance term ${{{}^V{k_{{\rm{d,z}}}}{\mathop{\rm sgn}} \left( {{}^V{v_z}} \right){}^Vv_z^2} \mathord{\left/
 {\vphantom {{{}^V{k_{{\rm{d,z}}}}{\mathop{\rm sgn}} \left( {{}^V{v_z}} \right){}^Vv_z^2} m}} \right.
 \kern-\nulldelimiterspace} m}$ in \eqref{eq:vcz}. And according to \cite{Ruiz-2022}, the least squares method can be 
 implemented to identifies the specific value of ${{{}^V{k_{{\rm{d,x}}}}} \mathord{\left/
  {\vphantom {{{}^V{k_{{\rm{d,x}}}}} m}} \right.
  \kern-\nulldelimiterspace} m}$.

Secondly, we need a reduced attitude which contains the information of $\Gamma_{xd}$, 
$\Gamma_{yd}$, and $\Gamma_{zd}$. 
Since the magnitude of $\Gamma_{yd}$ is relatively small,
we can use the normalization to straightforward compose the reduced attitude without significant accuracy loss:
\begin{align}
{\Gamma _{xp}} = {{{\Gamma _{xd}}} \mathord{\left/
 {\vphantom {{{\Gamma _{xd}}} {\sqrt {\Gamma _{xd}^2 + \Gamma _{yd}^2 + \Gamma _{zd}^2} }}} \right.
 \kern-\nulldelimiterspace} {\sqrt {\Gamma _{xd}^2 + \Gamma _{yd}^2 + \Gamma _{zd}^2} }}\\
{\Gamma _{yp}} = {{{\Gamma _{yd}}} \mathord{\left/
 {\vphantom {{{\Gamma _{yd}}} {\sqrt {\Gamma _{xd}^2 + \Gamma _{yd}^2 + \Gamma _{zd}^2} }}} \right.
 \kern-\nulldelimiterspace} {\sqrt {\Gamma _{xd}^2 + \Gamma _{yd}^2 + \Gamma _{zd}^2} }}\\
{\Gamma _{zp}} = {{{\Gamma _{zd}}} \mathord{\left/
 {\vphantom {{{\Gamma _{zd}}} {\sqrt {\Gamma _{xd}^2 + \Gamma _{yd}^2 + \Gamma _{zd}^2} }}} \right.
 \kern-\nulldelimiterspace} {\sqrt {\Gamma _{xd}^2 + \Gamma _{yd}^2 + \Gamma _{zd}^2} }}
\end{align}
where ${\Gamma _{xp}}$, ${\Gamma _{yp}}$, and ${\Gamma _{zp}} \in \mathbb R$, are the reduced attitude the attitude controller uses as desire.

Third, the low pass second order filters are used to generate derivative signals, such as ${{\dot {\bm v}}_d}$, ${{\dot \psi }_d}$, and ${{\dot \omega }_{\psi d}}$. This operation holds the same logic as the famous command filter \cite{Farrell-2009}, although auxiliary compensations are omitted here for brevity. In addition, there exist jumps in the signal of ${{\omega }_{\psi d}}$, the corresponding filter is reset once the jump happens, in avoiding 
unintended large derivative.

Finally,
the inner loop under-actuated attitude control problem should be considered.
In order to keep both the controller and the discussion simple, and 
based on the works in \cite{Reinhardt-2021}, further by neglecting the rotation motion around the $z$-direction of the
body-fixed frame, the following attitude control laws are implemented in the practical flight:
\begin{align}
{\theta _{{\rm{rud}}}} &= {k_{{\rm{rud}}}}\left( {{\Gamma _{yp}}{\Gamma _z} - {\Gamma _{zp}}{\Gamma _y}} \right) - k_{\omega,x} \omega_x\\
{\theta _{{\rm{ele}}}} &= {k_{{\rm{ele}}}}\left( {{\Gamma _{zp}}{\Gamma _x} - {\Gamma _{xp}}{\Gamma _z}} \right) - k_{\omega,y} \omega_y
\end{align}
where positive constants $k_{\rm rud}, k_{\rm ele} \in \mathbb R^+$ are the coefficients for the rudder and the elevator, respectively.
And the positive constants $k_{\omega,x}, k_{\omega,y} \in \mathbb R^+$ are the damping coefficients for the angular velocities along
the X-axis and the Y-axis.
In this controller, the assumption of the proportional relationship between torques and deflection angles is adopted.

\section{Experiments }
Real flight experiments are conducted on the FWAV to validate both the trajectory generation and trajectory tracking strategies. The trajectory generation process is accomplished offline using the MATLAB \emph{fmincon} nonlinear programming solver. Subsequently, the generated trajectories are fed into the online trajectory controller as the desire. The experiments are carried out with the assistance of the Qualisys motion capture arena, with 42 Qualisys Arqus A12 cameras online. And, as shown in Fig. \ref{figure:Overview}, a self-made X-wing FWAV is developed to achieve the real flight experiments, weighing approximately 29~g and possessing a wingspan of 34 cm. Two linear steering servos are utilized to actuate the elevator and rudder, respectively, while a brush-less motor is employed for the flapping wing motion. The control algorithm is executed on the ground station, which receives attitude and position information from the motion capture system. The control signals are then transmitted to the flying flapping robot using the ``DIY multi-protocol TX module''.
The overall system development is completed in our previous work \cite{Qian-2023}.

\begin{figure}[t]
      \centering
      \includegraphics[width=2.8in]{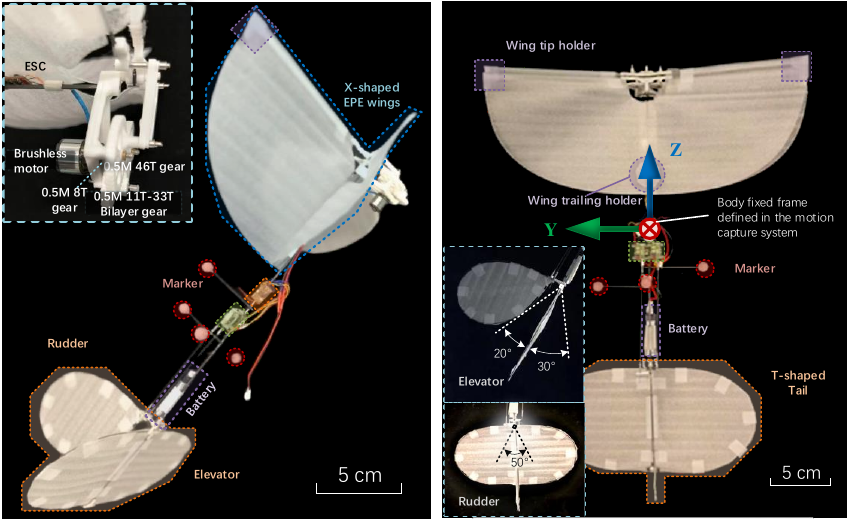}\\
      \caption{Overview of the flapping wing robot used in the real flight experiments.}
      \label{figure:Overview}
\end{figure}
\subsection{Different Flight Cases}
Let us then test the trajectory generation strategy by three carefully chosen cases:

{\bf Case~(a):} The initial position of the FWAV is ${\left[ {\begin{array}{*{20}{c}}
0&0&0 \end{array}} \right]^\top}~{\rm m} $ and the final position is ${\left[ {\begin{array}{*{20}{c}}
1&1&1 \end{array}} \right]^\top}~{\rm m}$. Both the initial and the final velocity are ${\left[ {\begin{array}{*{20}{c}}
0&0&0 \end{array}} \right]^\top}~{\rm m}$. There is also a ball obstacle. The ball center is located at ${\left[ {\begin{array}{*{20}{c}}
0.5 &0.5 &0.5 \end{array}} \right]^\top}~{\rm m}/s$, and the radius of the ball center is 0.5~m. 

{\bf Case~(b):} The initial position is ${\left[ {\begin{array}{*{20}{c}}
0&0&0 \end{array}} \right]^\top}~{\rm m}$ and the final position is ${\left[ {\begin{array}{*{20}{c}}
0&2&0 \end{array}} \right]^\top}~{\rm m}$. Both the initial and the final velocity are ${\left[ {\begin{array}{*{20}{c}}
0&0&0 \end{array}} \right]^\top}~{\rm m/s}$. In this scenario, there exist two axially unbounded cylindrical obstacles, each with a radius of 0.3~m.
The cylinders positions in $Y$-$Z$ plane are $\left[ {\begin{array}{*{20}{c}}{0.5}&{ - 0.2}\end{array}} \right]^\top~{\rm m}$ and 
$\left[ {\begin{array}{*{20}{c}}{1.5}&{  0.1}\end{array}} \right]^\top~{\rm m}$, respectively.

{\bf Case~(c):} In this particular scenario, a total of six waypoints are present, which are intended to be traversed by the trajectory.
 The initial position, the final position, and the 1st waypoint are all set as 
${\left[ {\begin{array}{*{20}{c}}1.5&0&0 \end{array}} \right]^\top}~{\rm m}$. Both the initial and the final velocity are ${\left[ {\begin{array}{*{20}{c}} 0&0&0 \end{array}} \right]^\top}~{\rm m/s}$. The remaining waypoints are sequentially given by  ${\left[ {\begin{array}{*{20}{c}}
{0.3~{\rm c} \frac{\pi }{3}}&{0.3~{\rm s}\frac{\pi }{3}}&0
\end{array}} \right]^\top}~{\rm m}$,  ${\left[ {\begin{array}{*{20}{c}}
{1.5~{\rm c} \frac{{2\pi }}{3}}&{1.5~{\rm s} \frac{{2\pi }}{3}}&0
\end{array}} \right]^\top}~{\rm m}$, ${\left[ {\begin{array}{*{20}{c}}
{-0.3~}&0&0
\end{array}} \right]^\top}~{\rm m}$, ${\left[ {\begin{array}{*{20}{c}}
{1.5~{\rm c} \frac{{4\pi }}{3}}&{1.5~{\rm s} \frac{{4\pi }}{3}}&0
\end{array}} \right]^\top}~{\rm m}$, ${\left[ {\begin{array}{*{20}{c}}
{0.3~{\rm c} \frac{{5\pi }}{3}}&{0.3~{\rm s} \frac{{5\pi }}{3}}&0
\end{array}} \right]^\top}~{\rm m}$.

\subsection{Trajectory Generation Results}
The actually used objective function also incorporates the 
velocity terms, such that it is formulated as
\begin{equation}
\label{eq:minsnap_pro}
\min \int_0^{M \cdot T} {\left( {{\mu _p}{{\sum\nolimits_{j = 1}^3 {\left\| {\frac{{{{\rm{d}}^4}{\sigma _j}}}{{{\rm{d}}{{\rm{t}}^4}}}} \right\|} }^2} }  + {{\mu _v}{{\sum\nolimits_{j = 1}^3 {\left\| {\frac{{{{\rm{d}}}{\sigma _j}}}{{{\rm{d}}{{\rm{t}}}}}} \right\|} }} }\right)} {\rm{dt}}
\end{equation}
where the weights $\mu_p$ and $\mu_v$ help balancing the trade-off between the minimum snap and the distance optimization.

The objective function in this context still adheres to a quadratic form, allowing for its analytical computation. Following the evaluation of various optimization strategies, the sequential quadratic programming algorithm is chosen as the preferred approach.  It is worth noting that other conventional optimization methods, such as the interior point method, can also yield similar outcomes.
The initial guesses of the polynomial coefficients are randomly generated within an appropriate range.
\begin{figure*}[t]
      \centering
      \includegraphics[width=5.5in]{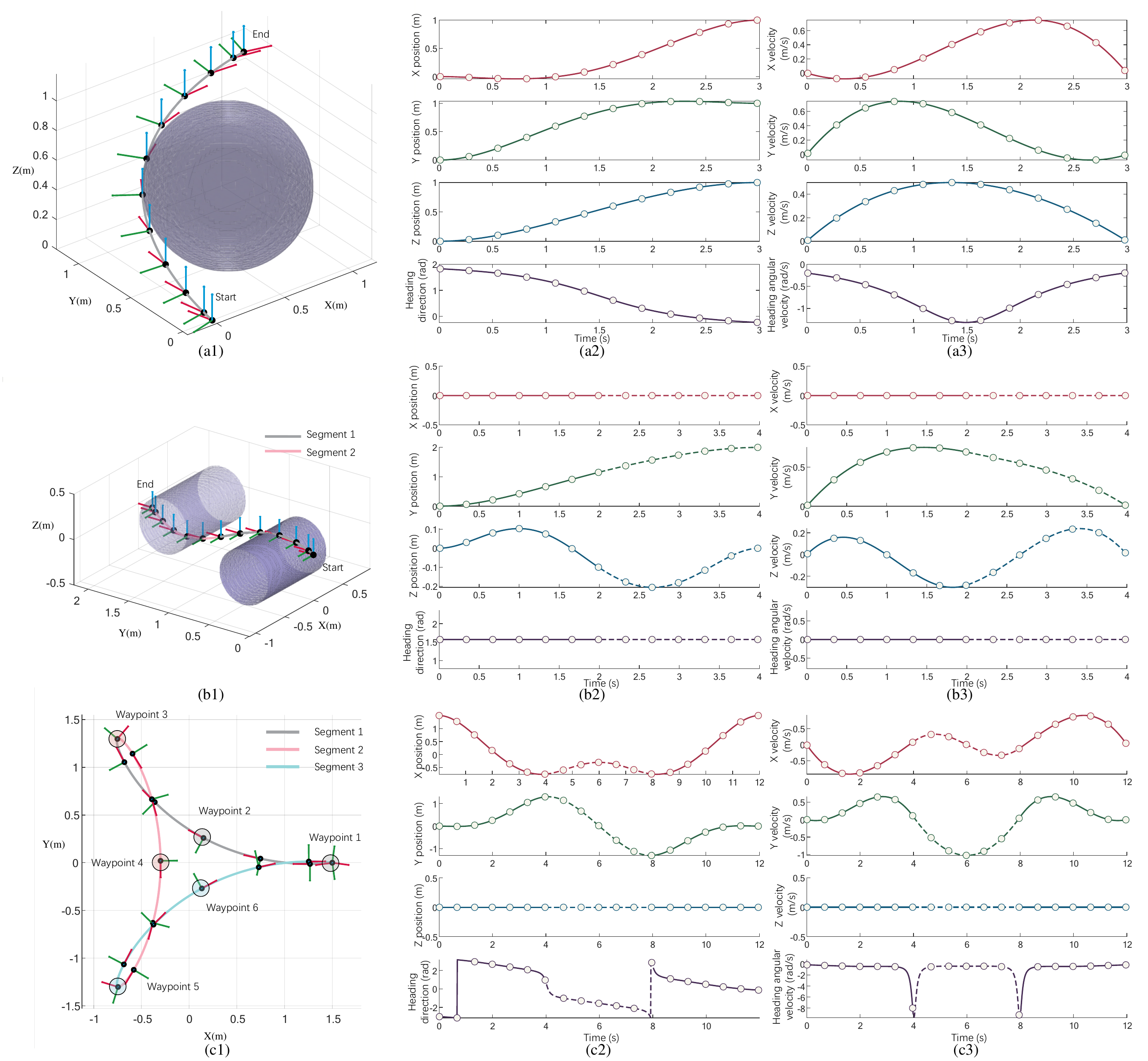}\\
      \caption{Generated trajectories when facing different constraints: (a1) single segment trajectory generation facing a ball obstacle,  (a2) the position curves of the 
      generated trajectory, (a3) the velocity curves of the generated trajectory. Sub-figures (b1), (b2), and (b3) exhibit trajectory generation involving two consecutive cylinder obstacles, using 2 segments of trajectories. On the other hand, sub-figures (c1), (c2), and (c3) demonstrate trajectory generation with six waypoint constraints, using 3 segments of trajectories. It is note worthing that the azimuth angles are not considered to be the trajectory components. The curves depicting this angle are illustrated as they aid in ascertaining the orientation of the vertical frames.
      }
      \label{figure:TrajGen}
\end{figure*}
Evidently, the aforementioned kinodynamic constraints, as well as the obstacle constraints, are imposed throughout the entirety of the trajectories. This is achieved by exerting the constraints on sampled points. 
Then the specific methodologies are described below with 
{\bf Case~(a)} as an example. 
The basic idea is to transform the constraint into a nonlinear function that is equal to zero.
First, let us consider the horizontal velocity constraint.
Here, the sampled time points are denotes as $\tau_i \in \left(0, 3\right)~{\rm s}$ with the mean sampling interval as $0.15~{\rm s}$. 
Therefore, the time points are shown as $\tau_1 = 0.15~{\rm s}$, $\tau_2 = 0.3~{\rm s}$, $\tau_3 = 0.45~{\rm s} \cdots$ $\tau_{19} = 2.85~{\rm s}$.
The corresponding equation is 
\begin{align}
&\sum\limits_{i = 1}^{20} {{\rm{Rec}}\left( {\sqrt {{{\dot x}^2} + {{\dot y}^2}}  - {v_{{\rm{h,max}}}}} \right)}  = 0\\
&{\mathop{\rm Re}\nolimits} c\left( \star \right) = \left\{ {\begin{array}{*{20}{c}}
{0,~~\star < 0}\\
{\star,~~\star \ge 0}
\end{array}} \right. \nonumber
\end{align} 
where the rectifying function ${\rm Rec}(\star) \in \mathbb R \to \mathbb R_{\ge 0}$ is utilized to ensure the equivalence for regions that are not part of the obstacle.
Second, the azimuth angular velocity constraint is shown as
\begin{equation}
\sum\limits_{i = 1}^{20} {{\rm{Rec}}\left( {\arctan\!2\left( {\dot x,\dot y} \right) - \dot \psi_{\max}} \right)}  = 0
\end{equation} 
Third, the ball obstacle constraint is shown as
\begin{equation}
\sum\limits_{i = 1}^{20} {{\rm{Rec}}\left( 0.5 - {\sqrt {{x^2} + {y^2} + {z^2}}  } \right)}  = 0
\end{equation} 

In addition, according to the trajectory complexity, we implement single segment, 2 segments, and 3 segments polynomials, respectively for {\bf Case~(a)},
{\bf Case~(b)}, and {\bf Case~(c)}. Each segment is configured as a 6th order polynomial function.
And in {\bf Case~(c)}, the intermediate points (waypoints 2, 4, and 6 shown in Fig. \ref{figure:TrajGen}) within each segments are assigned to the temporal midpoint of their respective durations, such that 
${\sigma _{1,1}}\left( {\frac{T}{2}} \right) = 0.3~{\mathop{\rm c}\nolimits} \frac{\pi }{3},~{\sigma _{1,2}}\left( {\frac{T}{2}} \right) = 0.3~{\rm s}\frac{\pi }{3},~{\sigma _{1,3}}\left( {\frac{T}{2}} \right) = 0
$.

Subsequent to these optimization configurations, the outcomes depicted in Fig. \ref{figure:TrajGen} are ultimately obtained \footnote{See source code at: https://github.com/Chainplain/CrappyMinimumSnap.}. Through a thorough examination of these results, our proposed strategy demonstrates its capability to generate trajectories that are smooth and safe across a wide range of scenarios.

% \subsection{Trajectory Generation Comparison}
% One notable innovation in our trajectory generation method lies in the integration of flapping wing dynamics. To assess its impact, we conducted trajectory generation experiments with different azimuth angular velocity constraint. 
% For the sake of conciseness and without loss of generality, the comparison analysis is solely conducted on {\bf Case~(a)}.

\subsection{Real Flight}
The generated trajectory is then implemented into real flapping wing flight to further validate the proposed strategies.
The FWAV is released from the operator hand initially.
Then the FWAV enters into the reduced attitude stabilizing mode. 
Subsequently, after an approximate duration of 10 seconds, the FWAV attains a stable flight state from the point of release.
Then the FWAV starts to track the generated trajectory, and simultaneously the inertia origin is reset to the location where the FWAV initiates its tracking, but its orientation remains unchanged. 
After the completion of the tracking task, the FWAV transitions back to the mode of reduced attitude stabilization.
The parameters of the controller are carefully fine-tuned to achieve the best possible outcome \footnote{See source code at: https://github.com/Chainplain/UATrajTrack. And see videos at: https://www.youtube.com/watch?v=yOA0aA6X4FI.}.

\begin{figure}[t]
      \centering
      \includegraphics[width=2.7in]{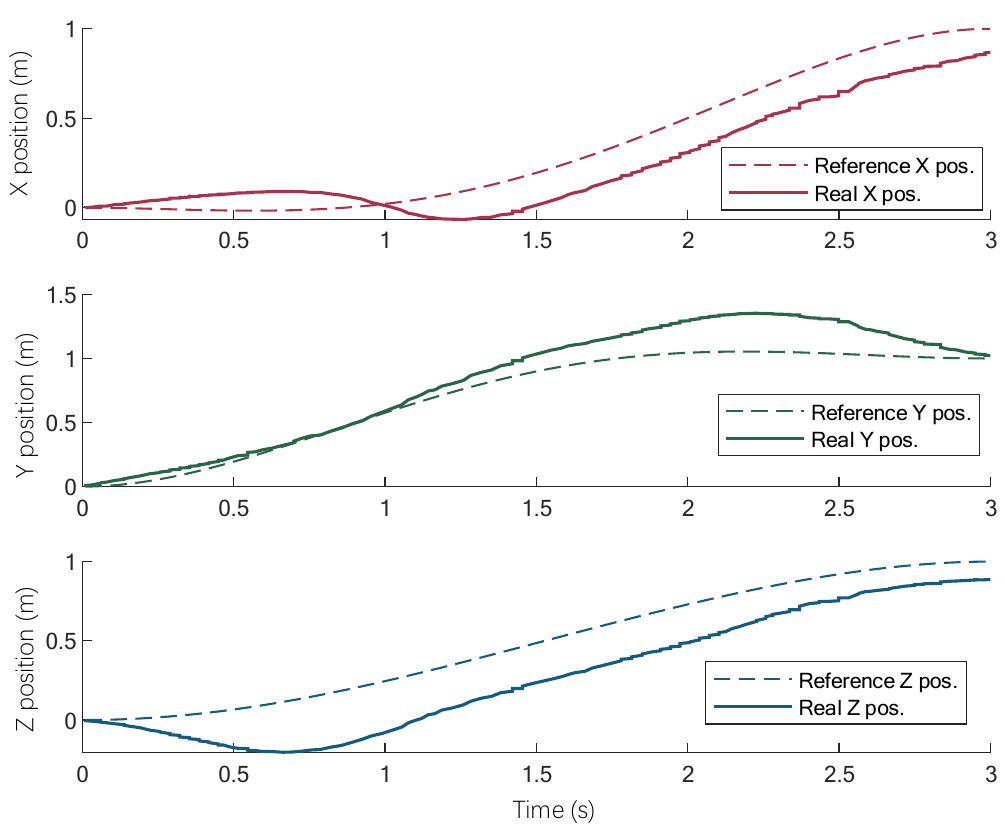}\\
      \caption{Real flight trajectory tracking results of {\bf Case~(a)}.}
      \label{figure:Ballres}
\end{figure}
\begin{figure}[t]
      \centering
      \includegraphics[width=2.7in]{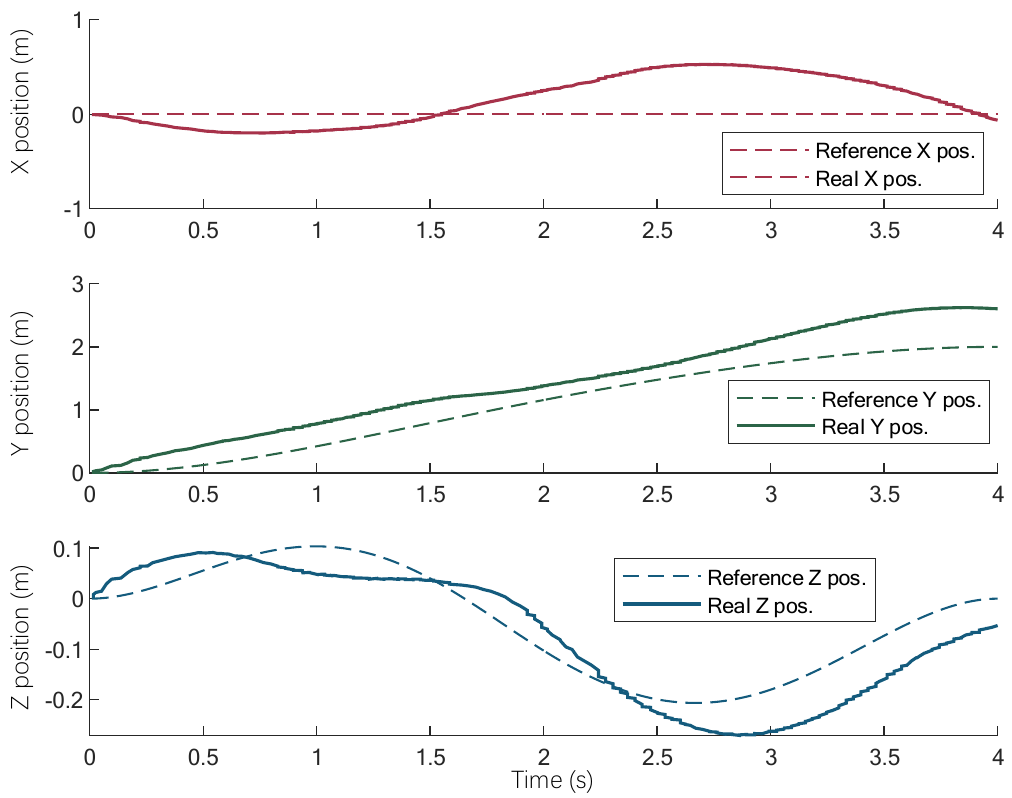}\\
      \caption{Real flight trajectory tracking results of {\bf Case~(b)}.}
      \label{figure:Wallres}
\end{figure}
The obtained experiment results of {\bf Case (a)} and {\bf Case (b)} are clearly shown in Fig. \ref{figure:Ballres}
and Fig. \ref{figure:Wallres}, respectively. 
The trajectory generation and control strategy proposed in this study allows the FWAV to produce feasible trajectories and subsequently navigate along them during flight. 
In fact, due to reasons of maintaining a stable flight, the FWAV is not limited to a condition of zero velocity when it begins the tracking tasks, which probably leads to the large deviations at the first few seconds.
Despite the presence of non-negligible position errors, the tracking behaviors exhibited by the FWAV are notably prominent, enabling effective avoidance of obstacles as defined in the trajectory generation process.
On the other side, the tracking performance for {\bf Case (c)} is not as satisfactory as {\bf Case (a)} and {\bf Case (b)}.
The experiment result behavior of {\bf Case (c)} is demonstrated  shown in Fig. \ref{figure:Wayres}.
It is obvious that, the FWAV can hardly follow the generated trajectory.
Based on the experiment result of {\bf Case (c)} ,
it becomes apparent that the larger positional error can be attributed to the rapid and frequent changes in heading direction, which emerges at waypoint~3 and waypoint~5 shown in Fig~\ref{figure:TrajGen}-(c1), as well as in Fig. \ref{figure:Wayres}.
In order to verify this proposition, we conduct additional experiments where the FWAV is expected to start at the origin and fly along the X-axis with the uniform speed of 0.5~m/s.
The obtained results are shown in Fig. \ref{figure:LineShow}.
As we can see, in the scenario where there is no changes in heading direction, the FWAV can relatively accurately track the desired trajectory.
This phenomenon underscores the significance of the heading-direction-changing-rate constraints in facilitating accurate and stable flight.
Nevertheless, incorporating this constraint into the nonlinear optimization process yields computational challenges that necessitate further exploration of the underactuated dynamics of the FWAV and its corresponding flight missions. Due to these complexities, this aspect is not encompassed within the scope of the present study.

\begin{figure}[t]
      \centering
      \includegraphics[width=2.7in]{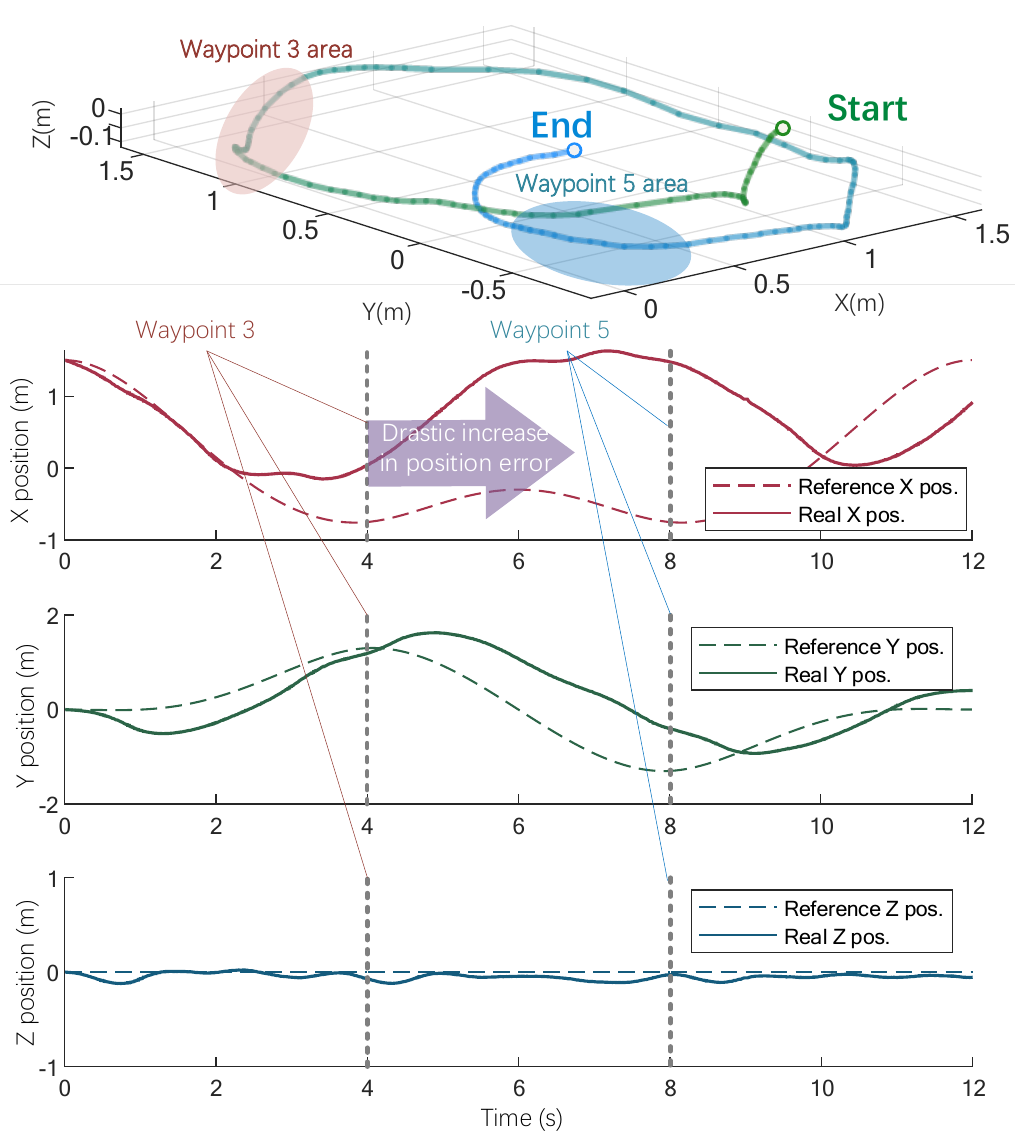}\\
      \caption{Real flight trajectory tracking results of {\bf Case~(c)}.}
      \label{figure:Wayres}
\end{figure}

\begin{figure}[t]
      \centering
      \includegraphics[width=2.7in]{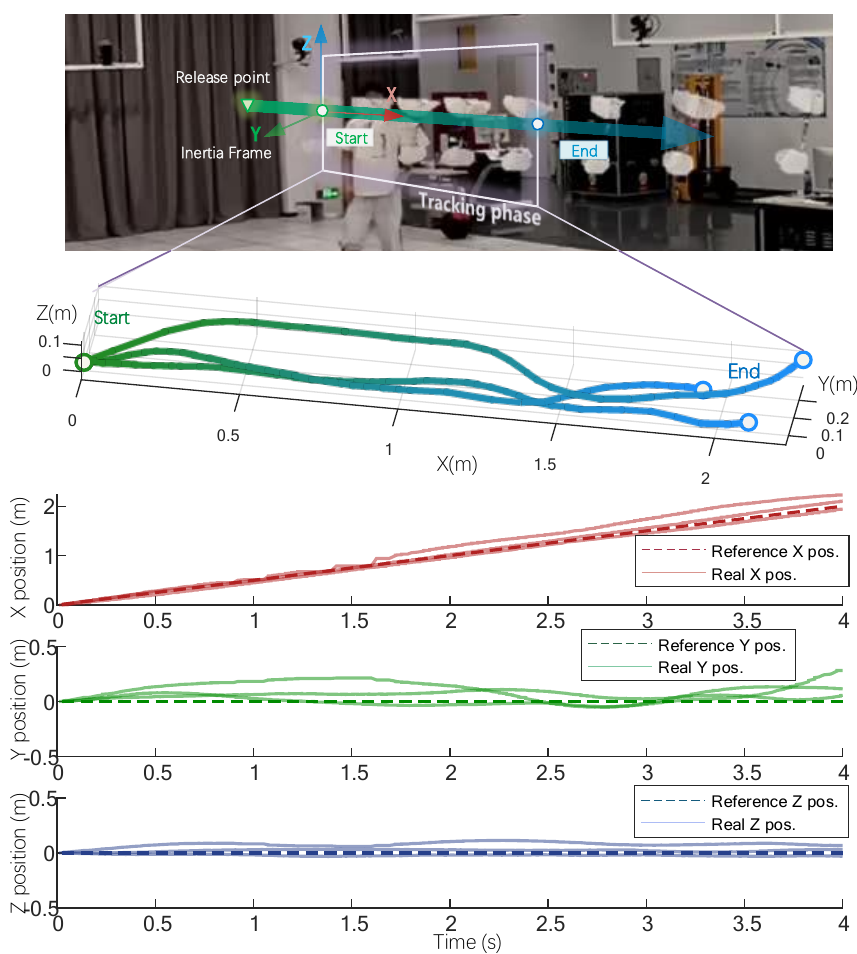}\\
      \caption{Real flight trajectory tracking results where the FWAV flies along a straight line with a constant speed of $0.5~{\rm m/s}$.}
      \label{figure:LineShow}
\end{figure}

Finally, the trajectory tracking control results are summarized in TABLE \ref{tab:Exp}. 
The errors are defined similar to those given in \cite{Hoffmann-2008}. The abbreviation MAX represents maximum, and RMS represents root mean square.
Based on these observations, it is evident that, during flight tasks that do not involve sharp climbs, the controller demonstrates its superior performance in altitude tracking. However, the along-track and cross-track errors, particularly the former, exhibit relatively non-negligible tracking errors. 
While compensating for these two errors, the FWAV requires an attitude adjustment process, which is further compounded by the intricate dynamics of the underlying aerodynamics.

\begin{table}[t]
\caption{Trajectory tracking experiment results in three directions}
  \centering
  \label{tab:Exp}
\begin{tabular}{clllllll}
\hline
\multicolumn{2}{c}{\multirow{2}{*}{Case}} & \multicolumn{2}{c}{\begin{tabular}[c]{@{}c@{}}Along-Track \\ Errors (m)\end{tabular}} & \multicolumn{2}{c}{\begin{tabular}[c]{@{}c@{}}Cross-Track \\ Errors (m)\end{tabular}} & \multicolumn{2}{c}{\begin{tabular}[c]{@{}c@{}}Altitude\\ Errors (m)\end{tabular}} \\ \cline{3-8} 
\multicolumn{2}{c}{}                      & \multicolumn{1}{c}{MAX}                   & \multicolumn{1}{c}{RMS}                   & \multicolumn{1}{c}{MAX}                   & \multicolumn{1}{c}{RMS}                   & \multicolumn{1}{c}{MAX}                 & \multicolumn{1}{c}{RMS}                 \\ \hline
\multicolumn{2}{c}{(a)}                   & 0.311                                     & 0.170                                     & 0.243                                     & 0.147                                     & 0.346                                   & 0.202                                   \\
\multicolumn{2}{c}{(b)}                   & 0.636                                     & 0.406                                     & 0.528                                     & 0.312                                     & 0.096                                   & 0.058                                   \\
\multicolumn{2}{l}{Lines}                 & 0.307                                     & 0.113                                     & 0.282                                     & 0.096                                     & 0.114                                   & 0.052                                   \\ \hline
\end{tabular}
\end{table}

\section{Conclusion}
In this paper, novel trajectory generation and tracking control strategies for an underactuated FWAV are proposed. 
The paper establishes the theoretical basis for trajectory planning, demonstrates the differential flatness property of the FWAV system, and develops a general-purpose trajectory generation strategy. 
A trajectory tracking controller is then proposed using robust and switch control techniques, ensuring overall system stability through Lyapunov analysis. 
The closed-loop integration of trajectory generation and control for real 3-dimensional flight in an underactuated FWAV is achieved. 
Although the trajectory generating and tracking are sufficient for common flight tasks, for example obstacle avoidance,
further improvements in accuracy are still anticipated to enhance the precision of flight missions.
In subsequent researches, we aim to modify the aerodynamic configuration of the FWAV while integrating both attitude and position estimation and control algorithms completely onboard. This endeavor is designed to achieve precise maneuverability and further to realize small-scale FWAV perching. 
Moreover, we anticipate incorporating the constraint of heading direction change rate into the optimization objective, while simultaneously maintaining a quadratic form. This modification aims to enhance the generation of stable and feasible trajectories in real-time.


% \end{IEEEbiography}


% % insert where needed to balance the two columns on the last page with
% %\newpage



%\vfill


\end{document}